\documentclass[10pt,twocolumn,twoside]{IEEEtran} 

\IEEEoverridecommandlockouts                              %

\usepackage{xr}

\newcommand{\omitted}[1]{}%

\title{\fontsize{21}{21} \selectfont 
Asynchronous Distributed Bandit Submodular Maximization under Heterogeneous Communication Delays
}
\author{Pranjal Sharma, Zirui Xu, Vasileios Tzoumas$^\dagger$
	\thanks{%
 $^\dagger$Department of Aerospace Engineering, University of Michigan, Ann Arbor, MI 48109, USA  {\tt\footnotesize \{spranjal,ziruixu,vtzoumas\}@umich.edu}} 
    \thanks{This work was supported by NSF CAREER Award No. 2337412 and ARO Early Career Program award W911NF-25-1-0280.}
}

\usepackage{cite}

\usepackage{comment}
\usepackage{siunitx}
\usepackage{relsize}
\usepackage{ifthen}
\usepackage[colorinlistoftodos]{todonotes}

\usepackage[vlined,ruled,linesnumbered]{algorithm2e}
\usepackage{graphics} %
\usepackage{rotating}
\usepackage{color}
\usepackage{enumerate}
\usepackage[T1]{fontenc}
\usepackage{psfrag}
\usepackage{epsfig} %
\usepackage{booktabs}
\usepackage{graphicx,url}
\usepackage{multirow}
\usepackage{array}
\usepackage{latexsym}
\usepackage{amsfonts}
\usepackage{amsmath}
\usepackage{amssymb}
\usepackage{xstring}
\usepackage{algorithmic}
\usepackage{multirow}
\usepackage{xcolor}
\usepackage{prettyref}
\usepackage{flexisym}
\usepackage{bigdelim}
\usepackage{breqn} %
\usepackage{listings}
\let\labelindent\relax
\usepackage{xspace}
\usepackage{bm}
\graphicspath{{./figures/}}
\usepackage{tikz}
\usetikzlibrary{matrix,calc}
\usepackage{lipsum}
\usepackage{mdwlist}

\let\itemize\undefined
\makecompactlist{itemize}{stditemize}
\usepackage{enumitem}
\usepackage{caption}
\usepackage{subcaption}

\usepackage{amsthm}
\usepackage{mathtools}
\usepackage{mleftright}

\makeatletter
\let\NAT@parse\undefined
\makeatother
\usepackage{hyperref}
\usepackage{cleveref}
\allowdisplaybreaks

\hypersetup{%
  colorlinks=true,
  linkcolor=blue,
  filecolor=magenta,      
  urlcolor=black,
  citecolor=red,
  linkbordercolor={0 0 1}
}

\newtheorem{theorem}{Theorem}
\newtheorem{problem}{Problem}

\newtheorem{lemma}{Lemma}
\newtheorem{assumption}{Assumption}
\newtheorem{definition}{Definition}

\newtheorem{remark}{Remark}

\newcommand{\bdmath}{\begin{dmath}}
\newcommand{\edmath}{\end{dmath}}
\newcommand{\beq}{\begin{equation}}
\newcommand{\eeq}{\end{equation}}
\newcommand{\bdm}{\begin{displaymath}}
\newcommand{\edm}{\end{displaymath}}
\newcommand{\bea}{\begin{eqnarray}}
\newcommand{\eea}{\end{eqnarray}}
\newcommand{\beal}{\beq \begin{array}{lll}}
\newcommand{\eeal}{\end{array} \eeq}
\newcommand{\beas}{\begin{eqnarray*}}
\newcommand{\eeas}{\end{eqnarray*}}
\newcommand{\ba}{\begin{array}}
\newcommand{\ea}{\end{array}}
\newcommand{\bit}{\begin{itemize}}
\newcommand{\eit}{\end{itemize}}
\newcommand{\ben}{\begin{enumerate}}
\newcommand{\een}{\end{enumerate}}

\newcommand{\calA}{{\cal A}}
\newcommand{\calB}{{\cal B}}
\newcommand{\calC}{{\cal C}}

\newcommand{\calE}{{\cal E}}

\newcommand{\calG}{{\cal G}}

\newcommand{\calN}{{\cal N}}

\newcommand{\calV}{{\cal V}}

\definecolor{myblue}{RGB}{65 105 225}

\newcommand{\hide}[1]{}

\newcommand{\hiddenText}{{\color{gray} hidden text.}}
\newcommand{\hideWithText}[1]{\hiddenText}

\newcommand{\opt}{^{\star}}

\DeclareRobustCommand{\scenario}[1]{%
  \ifmmode
    \mathsf{#1}%
  \else
    \texorpdfstring{{\fontsize{8.9}{9}\selectfont\sffamily #1}\xspace}{#1}%
  \fi
}

\newcommand{\ie}{{i.e.},\xspace}
\newcommand{\eg}{{e.g.},\xspace}
\newcommand{\myin}{\, \in \,}

\newcommand{\sg}{\scenario{SG}}
\newcommand{\banalg}{\scenario{BSG}}

\newcommand{\myParagraph}[1]{{\bf #1.}\xspace}

\renewcommand{\opt}{\scenario{OPT}}

\newcommand{\curv}{\kappa}
\newcommand{\ourcurv}{\scenario{coin}}

\newcommand{\alg}{\scenario{DOG-IU}}
\newcommand{\dog}{\scenario{DOG}}

\newcommand{\elem}{v}

\newcommand{\solopt}{\calA^{\opt}}

\newcommand{\Reg}{\operatorname{Reg}_{T}}

\PassOptionsToPackage{end}{algorithmic}

\crefname{assumption}{assumption}{Assumptions}

\IEEEaftertitletext{\vspace{-1\baselineskip}}

\begin{document}

\maketitle

\thispagestyle{empty}
\pagestyle{empty}

\begin{abstract}

We study asynchronous distributed decision-making for scalable multi-agent bandit submodular maximization.  We are motivated by distributed information-gathering tasks in unknown environments and under heterogeneous inter-agent communication delays. To enable scalability despite limited communication delays, existing approaches restrict each agent to coordinate only with its one-hop neighbors. But these approaches assume homogeneous communication delays among the agents and a synchronous global clock. In practice, however, delays are heterogeneous, and agents operate with mismatched local clocks. That is, each agent does not receive information from all neighbors at the same time, compromising decision-making. In this paper, we provide an asynchronous coordination algorithm to overcome the challenges. We establish a provable approximation guarantee against the optimal synchronized centralized solution, where the suboptimality gap explicitly depends on communication delays and clock mismatches. The bounds also depend on the topology of each neighborhood, capturing the effect of distributed decision-making via one-hop-neighborhood messages only. The clock-mismatch penalty is expressed through a bounded oscillation of the reward in the execution times, which, unlike a Lipschitz condition, is available for discrete-valued rewards such as target monitoring. We validate the approach through numerical simulations on multi-camera target monitoring through the lenses of two different reward structures and simulation setups.

\end{abstract}

\vspace{-3.3mm}
\section{Introduction}\label{sec:Intro}
Multi-agent systems of the future will increasingly rely on agent-to-agent communication to coordinate tasks such as target tracking~\cite{xu2023bandit}, environmental mapping~\cite{atanasov2015decentralized}, and area monitoring~\cite{corah2018distributed}. These tasks are often modeled as 
\begin{equation}\label{eq:intro}
    \vspace{-.5mm}
	\max_{a_{i,t}\,\in\,\mathcal{V}_i,\,  \forall\, i\,\in\, \calN}\
	f_t(\,\{a_{i,t}\}_{i\myin \calN}\,), \;\;\;t=1,2,\dots,\vspace{-.5mm}
\end{equation}
across the robotics, control, and machine learning communities, where $\calN$ denotes the set of agents, $a_{i,t}$ denotes agent $i$'s chosen action at time $t$, $\calV_i$ denotes agent $i$'s set of available actions, and $f_t\colon 2^{\prod_{i \in \calN}\calV_i}\mapsto\mathbb{R}$ denotes the objective function that captures the task utility (global objective)~\cite{krause2008near,singh2009efficient,tokekar2014multi,atanasov2015decentralized,gharesifard2017distributed,marden2017role,grimsman2019impact,corah2018distributed,schlotfeldt2021resilient,du2022jacobi,rezazadeh2023distributed,robey2021optimal,xu2025communication}. In resource allocation and information gathering applications, $f_t$ is \textit{submodular}~\cite{fisher1978analysis},  a diminishing-returns property~\cite{krause2008near}. For example, in target monitoring with multiple reorientable cameras, $\calN$ is the set of cameras, $\calV_i$ represents the possible orientations of each camera, and $f_t$ measures the number of distinct targets observed within the joint field of view.

The optimization problem in~\cref{eq:intro} is NP-hard~\cite{Feige:1998:TLN:285055.285059}, but polynomial-time algorithms with provable approximation guarantees exist when the $f_t$ is submodular. A classical example is the \emph{Sequential Greedy} (\sg) algorithm~\cite{fisher1978analysis}, which guarantees a $1/2$-approximation ratio. 
\sg and its variants have been widely adopted in the controls, machine learning, and robotics literature~\cite{krause2008near,singh2009efficient,tokekar2014multi,atanasov2015decentralized,gharesifard2017distributed,grimsman2019impact,corah2018distributed,schlotfeldt2021resilient,liu2021distributed,robey2021optimal,rezazadeh2023distributed,konda2022execution,krause2012submodular,xu2025communication}.


In this paper, we consider settings where the dynamics of the environment are unknown and partially observable. This requires agents to optimize actions based on retrospective rewards only (bandit optimization~\cite{lattimore2020bandit}). For example, in target tracking with unknown target motion~\cite{sun2020gaussian}, agents cannot evaluate $f_t$ in advance and instead rely on bandit feedback~\cite{lattimore2020bandit}, observing only the rewards of executed actions. This severely limits information reuse and complicates coordination. To address this, prior work extends sequential greedy to the bandit setting~\cite{xu2023online,xu2023bandit}, leveraging tools from online learning such as tracking the best expert (e.g., \scenario{EXP3-SIX}~\cite{neu2015explore}) to obtain suboptimality guarantees relative to time-varying optimal actions in hindsight.

However, the approaches above, similar to their offline counterparts~\cite{krause2008near,singh2009efficient,tokekar2014multi,atanasov2015decentralized,gharesifard2017distributed,grimsman2019impact,corah2018distributed,schlotfeldt2021resilient,liu2021distributed,robey2021optimal,rezazadeh2023distributed,konda2022execution,krause2012submodular}, where $f_t$ is assumed known \textit{a priori}, rely on sequential multi-hop communication over connected networks, leading to prohibitive delays under realistic communication constraints~\cite{xu2025communication}. Specifically, their communication complexity scales quadratically or cubically with the number of agents, and convergence typically requires a quadratic number of decision rounds. For instance, Bandit Sequential Greedy (\banalg)~\cite{xu2023bandit} incurs cubic communication per round and quadratic rounds to converge, resulting in quintic time complexity in the worst case~\cite[Theorem~6]{xu2026self}. To improve scalability, recent distributed approaches restrict coordination to one-hop neighbors and operate over arbitrary network topologies, achieving linear-time scaling. For example, Resource-Aware distributed Greedy (\scenario{RAG})~\cite{xu2025communication} matches centralized performance offline under full connectivity but incurs topology-dependent suboptimality otherwise. \cite{xu2026self} extends \scenario{RAG} to the online setting and actively designs each agent's communication neighborhood to maximize the overall optimization performance. Moreover, multi-hop communication is leveraged in~\cite{xu2026distributed} such that the coordination performance can be improved without sacrificing much decision speed. 

But all works above make two key assumptions: (i) they assume homogeneous one-hop communication delays among all agents, and (ii) they assume synchronized global clocks for all agents. These assumptions are crucial in enabling both the algorithms and theoretical guarantees for the prior works. But in practice, delays are generally heterogeneous across neighbors because of nonuniform communication hardware and local channel conditions; hence, information arrives at different times and decisions may have to be made before all information arrives. Moreover, the agents' local clocks are generally mismatched, and the multi-agent system cannot reliably maintain strict global synchronization. After incorporating these two limitations, the following research question arises:
\textit{How does each agent perform scalable coordination with others using partial-neighborhood information and under asynchronous local clocks?}

\myParagraph{Contributions}
We provide a distributed multi-agent decision-making framework that enables near-optimal action coordination in unknown environments under heterogeneous communication delays and asynchronous local clocks. 
Our approach leverages heterogeneous delays to allow each agent to incorporate partial neighborhood information as it arrives, allowing agents to learn near-optimal actions and adapt to dynamic environments faster. 
To this end, we develop tools for \textit{adversarial bandit with delayed feedback and asynchronous distributed submodular maximization}. The approach is fully distributed: each agent has its own pace of action selection under asynchronous local clocks.
We verify the algorithm's performance through multi-camera target-tracking simulations, showing that it increasingly outperforms the baseline as delays increase. The algorithm has the following properties:

\setcounter{paragraph}{0}
	
\paragraph{Approximation Performance}
The algorithm enjoys a suboptimality bound against the optimal solution of \cref{eq:intro}. In the synchronous setting, the bound captures the suboptimality gap against the optimal synchronized centralized solution, where the gap explicitly depends on communication delays and the topology of each neighborhood, capturing the effect of distributed decision-making via one-hop-neighborhood messages only (\Cref{thm:approx_performance}). In the asynchronous setting, given a timing mismatch bound of $\rho$, these guarantees remain valid up to an additive mismatch term of order
$O(|\mathcal{N}|\,\widetilde{\mathrm{osc}}(\rho))$, where $\widetilde{\mathrm{osc}}(\rho)$ bounds how much the reward can oscillate when execution times are perturbed by at most $\rho$ (\Cref{thm:async}). The term is linear in the number of agents, vanishes with $\rho$ whenever the reward is continuous in the execution times, and otherwise approaches a floor set by the reward's discontinuities.

\paragraph{Convergence Rate} 
The algorithm enables the agents to achieve epsilon-convergence after $\tilde{O}\!\left({|\mathcal{N}|^2 |\bar{\mathcal V}| \bar{M}_t}/{\varepsilon^2}\right)$ rounds, assuming the delays are bounded due to sufficient communication bandwidth. 



\section{Distributed Online Submodular Maximization Under Heterogeneous Communication Delays}\label{sec:problem}

We present the problem formulation. To this end, we use the following notation:
\begin{itemize}[leftmargin=*]
    \item $\calV_\calN \triangleq \prod_{i\myin \calN} \,\calV_i$ is the cross product of sets $\{\calV_i\}_{i\myin \calN}$.
    \item $[T]\triangleq\{1,\dots,T\}$ for any positive integer $T$;
    \item $f(\,a\,|\,\calA\,)\triangleq f(\,\calA \cup \{a\}\,)-f(\,\calA\,)$ is the marginal gain of set function $f:2^\calV\mapsto \mathbb{R}$ for adding $a \in \calV$ to $\calA \subseteq\calV$.
    \item $|\calA|$ is the cardinality of a discrete set $\calA$. 
\end{itemize}

We also use the following framework about the agents' communication network and their global objective $f$.

\myParagraph{Communication network} 
The distributed communication network $\calG=\{\calN, \calE\}$ \textit{can be directed and even disconnected}, where $\calE$ is the set of communication channels. When $\calG$ is fully connected (all agents receive information from all others), we call it \textit{fully centralized}. In contrast, when $\calG$ is fully disconnected (all agents are isolated, receiving information from no other agent), we call it \textit{fully decentralized}.

\myParagraph{Communication neighborhood}  
When a communication channel exists from agent $j$ to $i$, \ie $(j\rightarrow i) \in \calE$, $i$ can receive, store, and process information from $j$. The set of all agents from which $i$ can receive information through one-hop communication is denoted by $\calN_i$, agent $i$'s \textit{neighborhood}. We assume $\calN_i$ to remain constant over $[T]$. Information originating from different neighbors $j\in\calN_i$ may take varying amounts of time to reach $i$, depending on the message size and communication data rate.

\textbf{Communication delay.} For information sent from agent $j$ to agent $i$ at round $t$, let $d_{i,t}^j$ denote the communication delay. These delays may vary across neighbors $j\in\mathcal N_i$ and across time, reflecting heterogeneous communication conditions. Hence, agent $i$ can evaluate the reward of its round-$t$ action only after receiving the required neighbor actions, i.e., after a delay of $\max_{j\in\mathcal N_i} d_{i,t}^j$. We also define an upper bound on the delays for agent $i$ as $\bar{d}_i \triangleq \max_t(\max_{j \in \mathcal{N}_i}(d^j_{i,t}))$ and an upper bound on delays throughout the network as $\bar{d} \triangleq \max_{t}(\max_{i \in \mathcal{N}}(\max_{j \in \mathcal{N}_i}d^{j}_{i,t}))$. To this end, we also assume sufficient bandwidth for each communication channels such that the delays are bounded instead of accumulating.

\textbf{Arrival of information.}
Since the delays $d_{i,t}^j$ may differ across $j\in\mathcal N_i$, the round-$t$ neighbor actions received by agent $i$ may arrive in $K$ batches, where $1\le K\le |\mathcal N_i|$. 
\begin{align}
R_{i,t}^{(k)} &:= \{j:\text{ $a_{j,t}$  info. has arrived by batch }k\},\\
R_{i,t}^{(0)} &:= \emptyset,\\
M_{i,t}^{(k)} &:= \mathcal N_i \setminus R_{i,t}^{(k)}, \quad k \in \{1,\ldots,K\}.
\end{align}

\textbf{Reward Estimation.}
The agents may build estimates of neighbors' missing actions $M^{(k)}_{i,t}$ based on the neighbors' past actions and states. This allows the agents to estimate each round's reward before the true value can be computed. In our simulations, all agents use the last known neighbor actions as estimates for missing actions. \Cref{thm:per_agent} also covers regret guarantees for worst case estimates, which is when the difference between estimated and true reward is the maximum. This is possible since we assume the reward function to be bounded: such a conservative bound is available based on knowledge of worst case dynamics of the agents and the environment, a typical assumption for bandit learning \cite{thune2019nonstochastic}.

\begin{definition}[Normalized and Non-Decreasing Submodular Set Function{~\cite{fisher1978analysis}}]
\label{def:submodular}
A \emph{set function} $f:2^{\mathcal{V}}\mapsto\mathbb{R}$ is normalized and non-decreasing submodular \emph{if and only if}
\begin{itemize}[leftmargin=*]
\item (Normalization) $f(\,\emptyset\,)=0$;
\item (Monotonicity) $f(\,\calA\,)\leq f(\,\calB\,)$, $\forall\,\calA\subseteq \calB\subseteq \calV$;
\item (Submodularity) $f(\,s\,|\,\calA\,)\geq f(\,s\,|\,{\mathcal{B}}\,)$, $\forall\,\calA\subseteq {\mathcal{B}}\subseteq\calV$ and $s\in \calV$.
\end{itemize}
\end{definition}


\begin{definition}[2nd-order Submodular Set Function{~\cite{crama1989characterization,foldes2005submodularity}}]
$f:2^{\mathcal{V}}\mapsto\mathbb{R}$ is 2nd-order submodular if and only if
\begin{equation}
f(s\,|\,\calC) - f(s\,|\,\calA\cup\calC) \geq f(s\,|\,\calB\cup\calC) - f(s\,|\,\calA\cup\calB\cup\calC),\label{eq:second_order_submod}
\end{equation}
for any disjoint $\calA, \calB, \calC\subseteq\mathcal{V}$ $(\calA\cap \calB\cap \calC=\emptyset)$ and $s\in\mathcal{V}$.
\end{definition}


\begin{problem}[Distributed Online Submodular Maximization under Communication Delays]
\label{prob:DOSM}
At each time step $t \in [T]$, each agent $i \in \mathcal{N}$, given its neighborhood $\mathcal{N}_i$, needs to select an action $a_{i,t}$ to jointly solve
\begin{equation}
\label{eq:problem1_obj}
    \max_{a_{i,t} \in \mathcal{V}_i,\, \forall i \in \mathcal{N}}
    \; \sum_{t=1}^{T} f_t\big(\{a_{i,t}\}_{i \in \mathcal{N}}\big),
\end{equation}
where $f_t : 2^{\mathcal{V}^{\mathcal{N}}} \to \mathbb{R}$ is a normalized, non-decreasing submodular, and 2nd-order submodular set function, and each agent $i$ can access the value of $f_t(\mathcal{A})$ only after it has selected $a_{i,t}$ at time $t$ and received $\{a_{j,t}\}_{j \in \mathcal{N}_i}$ at time $t + d^j_{i,t}, \, \forall \mathcal{A} \subseteq \{a_{i,t}\} \cup  \{a_{j,t}\}_{j \in \mathcal{N}_i}$.
\end{problem}


\Cref{prob:DOSM} is the same as the one presented in~\cite{xu2026distributed} with an additional consideration for when the action data from an agent's neighbors is received. \Cref{prob:DOSM} also highlights a tradeoff: larger coordination neighborhoods can improve action quality, but they also increase the delay before an agent can evaluate its reward, since that reward depends on neighbors' round-$t$ actions. To avoid waiting for all missing information, we adopt an estimation-correction approach in which each agent forms intermediate reward estimates using the currently received neighbor actions and refines them as additional information arrives. This motivates the delayed-bandit formulation in the next section.

\vspace{-1mm}\section{Distributed Online Greedy with Intermediate Updates Algorithm (\alg)} \label{sec:algorithm}

\setlength{\textfloatsep}{3mm}
\begin{algorithm}[t]
\caption{Distributed Online Greedy with Intermediate Updates (\alg) for Agent $i$
}
\label{alg:async_DOG}
\begin{algorithmic}[1]
\REQUIRE Number of time steps $T$, agent $i$'s action set $\mathcal{V}_i$,
         agent $i$'s in-neighborhood $\mathcal{N}_i$, communication delay bound $\bar{d}_i$.
\ENSURE Agent $i$'s action $a_{i,t}$, $\forall t \in [T]$.
\STATE $\eta_i \gets \sqrt{\log |\mathcal{V}_i| \big/ \big((|\mathcal{V}_i| + \bar{d}_i)T\big)}$;
\STATE $w_1 \gets [w_{v,1},\ldots,w_{|\mathcal{V}_i|,1}]^{\top}$ with $w_{v,1} = 1$,
       $\forall a \in \mathcal{V}_i$;
\FOR{each time step $t \in [T]$}
  \STATE get distribution $p_t \gets w_t / \|w_t\|_1$;
  \STATE draw action $a_{i,t} \in \mathcal{V}_i$ from $p_t$;
  \STATE broadcast $a_{i,t}$ to one-hop neighbors;
  \STATE receive neighbors' actions $\{a_{j,s}\}_{j \in \mathcal{N}_{i,s}}$ for all $s \in \mathcal{S}_t \triangleq \{s: s + d^{j}_{i,s} = t\}$;
  \STATE form estimates $Z^{t}_0 \text{ and } Z^{s}_{k_s}, \forall s \in \mathcal{S}_t$;
  \STATE $\hat r_{a,s}( Z^{s}_{k_s}) \gets 1 - \dfrac{\mathbf{1}(a_{i,s} = a)}{p_{a,s}}
         \bigl(1 - Z^{s}_{k_s}\bigr),$ $ \forall a \in \mathcal{V}_i,  \forall s \in \mathcal{S}_t\cup\{t\}$;
  \STATE form corrections $\Delta_{a,s}, \; \forall a \in \mathcal{V}_i, \forall s \in \mathcal{S}_t\cup\{t\}$;
    
    \STATE $w_{a,t+1} \gets w_{a,t} \exp\left( \sum_{s \in \mathcal{S}_t \cup \{t\}} {\Delta}_{a,s}\right), \; \forall a \in \mathcal{V}_i$;
  \STATE store all $Z^{s}_{k_s}$;
\ENDFOR
\end{algorithmic}
\end{algorithm}

We present the Distributed Online Greedy with Intermediate Updates algorithm (\alg) for \Cref{prob:DOSM}. Particularly, \Cref{prob:DOSM} takes the form of adversarial bandit problems with delayed feedback. However, we also need to enable intermediate updates using partial information (from a subset of an agent's neighborhood). Therefore, we generalize the adversarial bandit with delayed feedback problem formulation to allow for intermediate updates (\Cref{subsec:adversarial_bandit}), and then present the main algorithm (\Cref{subsec:algorithm}).

\subsection{Per-Agent Adversarial Bandit with Delayed Feedback and Intermediate Updates}
\label{subsec:adversarial_bandit}
The adversarial bandit with delayed feedback problem involves an agent selecting a sequence of actions to maximize the total reward over a given number of time steps~\cite{thune2019nonstochastic}. The challenges are: (i) at each time step $t$, no action's reward is known to the agent a priori, and (ii) after an action is selected, only the selected action's reward will become known with a time delay $d_t$, which is assumed to be known a priori. We present the problem in the following using the notation:
\begin{itemize}[leftmargin=*]
  \item $\mathcal{V}$ denotes the available action set;
  \item $v_t \in \mathcal{V}$ denotes the agent's selected action at $t$;
  \item $r_{v_t,t} \in [0,1]$ denotes the reward of selecting $v_t$ at $t$, which in our case is a submodular function marginal. In other words, the agent's reward is the marginal gain of its action $v_t$ given the actions of its neighbors;
  \item $d_t$ is the number of delayed time steps for the reward of selecting action $v_t$ at $t$ to be received. In our case, the agent will know the actions of all of its neighbors and be able to calculate $r_{v_t,t}$ at $t+d_t$;
  \item \textbf{Intermediate estimates:} From $t$ until $t + d_t$, the agent will form estimates of the round $t$ reward as it receives more round $t$ information.
\end{itemize}

\begin{problem}[Adversarial Bandit with Delayed Feedback and Intermediate Updates]\label{prob:ABD}
\emph{Consider a horizon of $T$ time steps. At each time step $t \in [T]$, the agent $i$ needs to select an action $v_t \in \mathcal{V}$ such that the regret}
\begin{equation}\label{eq:abd_regret}
\mathrm{Regret}_T \triangleq \max_{v \in \mathcal{V}} \sum_{t=1}^{T} r_{v,t}
\;-\; \sum_{t=1}^{T} r_{v_t,t},
\end{equation}
\emph{is minimized, where no actions' rewards are known \emph{a priori}, and only the selected action's \textbf{true} reward $r_{v_t,t} \in [0,1]$ will become known at $t+ d_t$, with partial information about the reward becoming available in multiple batches at intermediate rounds between $t$ and $t + d_t$.}
\end{problem}

This problem is a more general version of the delayed bandit feedback problem tackled by the Delayed Exponential Weights (\scenario{DEW}) algorithm in \cite{thune2019nonstochastic} as it allows for intermediate updates based on estimates of missing rewards using partial information. In the case of all of the delayed information for a round arriving at the same time, \Cref{prob:ABD} reduces to the delayed bandit feedback problem discussed in \cite{thune2019nonstochastic}. The goal of solving Problem~\ref{prob:ABD} is to achieve a sublinear $\mathrm{Regret}_T$, i.e., $\mathrm{Regret}_T / T \to 0$ for $T \to \infty$, since this implies that the agent asymptotically chooses optimal actions even though the rewards are unknown a priori.


\subsection{\alg Algorithm}
\label{subsec:algorithm}
We enable agents in the distributed setting to solve \Cref{prob:DOSM} by making them simultaneously solve their own instance of \Cref{prob:ABD}. Intuitively, our goal is for each agent $i$ at each time step to efficiently select an action $a_{i,t}$ that maximizes the marginal gain $f_t(a_{i,t}\, |\, \{a_{j,t}\}_{j\in \mathcal{N}_i})$ from the perspective of agent $i$. Thus, \alg aims to efficiently minimize the following quantification:
\begin{definition}[Static Regret for Each Agent $i$]
Given that agent $i$ has a neighborhood $\mathcal{N}_i$, and at each time step $t$, agent $i$ selects an action $a_{i,t}$.
Then, the static regret of $\{a_{i,t}\}_{t\in[T]}$ is defined as
\begin{equation}
\begin{aligned}
\mathrm{Reg}_T\!\left(\{a_{i,t}\}_{t\in[T]}\right)
&\triangleq
\max_{a\in \mathcal{V}_i}\ \sum_{t=1}^{T} 
f_t\!\left(a \mid \{a_{j,t}\}_{j\in\mathcal{N}_i}\right)\\ 
&\quad-\sum_{t=1}^{T} 
f_t\!\left(a_{i,t} \mid \{a_{j,t}\}_{j\in\mathcal{N}_i}\right).
\end{aligned}
\label{eq:static_regret_agent_i}
\end{equation}
\end{definition}


Because the neighbors' round-$t$ actions arrive with heterogeneous delays, agent $i$ cannot evaluate the true reward $r_{a_{i,t}}\triangleq f_t\bigl(a_{i,t}\,|\, \{a_{j,t}\}_{j\in\mathcal{N}_i}\bigr)$ immediately after selecting $a_{i,t}$. Instead, \alg forms an intermediate estimate of this reward using the actions already received for round $t$ together with estimates of the still-missing neighbor actions:
\begin{equation}
    Z^t_k \triangleq f(a_{i,t} \; | \; \big\{a_j\big\}_{j \in R_{i,t}^{(k)}} \cup \big\{\tilde{a}_j\big\}_{j \in M_{i,t}^{(k)}}),
\end{equation}
where $k\in\{0,1,\dots,K\}$ is the number of information batches for round $t$ received so far. In particular, $Z_K^t=r_{a_{i,t},t}$ once all neighbors' round-$t$ actions have arrived.




Following the standard \scenario{EXP3} approach, agent $i$ uses the importance weighted estimate
\begin{equation}
\label{eq:imp_weight_estimator}
    \hat{r}_{a,t}(x) \triangleq 1 - \frac{\mathbf{1}(a = a_{i,t})}{p_{a,t}}(1-x),
\end{equation}
where $x$ is either an intermediate estimate $Z_t^k$ or the true reward. At round $t$, agent $i$ maintains, for each unresolved past round $s$, the currently received and still-missing neighbor sets, and refines its estimate whenever new information for that round arrives. Let $k_s$ denote the number of batches for round $s$ received up to round $t$. \alg then applies
\begin{align}
    &{\Delta}_{a,t} = \eta_i\, \hat r_{a,t}( \hat{Z}^{t}_0), \label{eq:updates_1}
    \\
    &{\Delta}_{a,s} = \eta_i \left[ \hat r_{a,s}\bigl({Z}^{s}_{k_s}\bigr)
  - \hat r_{a,s}\bigl({Z}^{s}_{k_s-1}\bigr)\right], \label{eq:updates_2}
  \\ 
  &\notag \forall a \in \mathcal{V}_i, \quad s \in \{t-\bar{d}_i,\dots, t-1\},
\end{align}where $\eta_i$ is the learning rate. $\Delta_{a,t}$ is the update made after agent $i$ acts at round $t$, while $\Delta_{a,s}$ is a correction applied when additional round $s$ neighbor information arrives and refines the reward estimate for that round.

\Cref{alg:async_DOG} implements this procedure online. It initializes the learning rate and action weights (lines 1--2). Then at each round it computes the sampling distribution and draws an action (lines 4--5), broadcasts chosen action and receives newly arrived delayed neighbor actions (lines 6--7), forms updated reward estimates for the current and unresolved past rounds (line 8), converts them into importance-weighted estimates and corrections (lines 9--10), and finally updates the weights (line 11) before storing the new estimates (line 12).

\section{Guarantees}\label{sec:bound}

We present the static regret bound of \alg's per-agent solution to Problem~\ref{prob:ABD}. Then, we present the suboptimality bound of \alg at the network level. The bound compares \alg's solution to the optimal solution of \Cref{prob:DOSM}. Leveraging the concept of \scenario{coin} (\Cref{def:coin}) that captures the suboptimality cost of distributed communication and computation, the bound covers the spectrum of \alg's approximation performance from when the network is fully centralized (all agents communicating with \textit{all}) to fully decentralized (all agents communicating with \textit{none}). Finally, we present the convergence analysis of \alg.

\begin{definition}[Cumulative Error]
For each round $t$, we define the cumulative error of \alg's reward estimates compared to the true rewards for rounds $s \in \{t-\bar{d}_i +1, \ldots, t\}$ as
\begin{equation}
\label{eq:cumulative_err_def}
    \varepsilon_{a,t}^i \triangleq \sum_{s = t-\bar{d}_i +1}^{t} \hat{r}_{a,s}(Z^{s}_{k_s}) - \hat{r}_{a,s}(r_{a_{i,s}}),
\end{equation}
where $r_{a_{i,s}}$ is the true reward for agent $i$'s action for round $s$ and $Z^s_{k_s}$ is $i$'s current estimate of the round $s$ reward.

We also define the maximum cumulative error over the action set as
\begin{equation}
    M_t^i \triangleq \max_{a \in \mathcal{V}_i} \,|\varepsilon^i_{a,t}|. \label{eq:Mt_definition}
\end{equation}
\end{definition}

$M^i_t$ is the worst-case absolute error (across actions) in the cumulative loss estimates at round $t$. 

\begin{definition}[Average Maximum Cumulative Error]
    For a horizon of $T$ rounds, define
    \begin{equation}
    \label{eq:Mt_bar_definition}
        \bar{M}_T^i \triangleq \frac{1}{T}\sum_{t=1}^{T} \mathbb{E}[M_t^i].
    \end{equation}
\end{definition}

$\bar{M}_T^i$ is a measure of how far \alg's internal model of rewards deviates from the true importance weighted rewards on average over the horizon for agent $i$.

\begin{theorem}[Per-Agent Adversarial Bandit with Delayed Feedback and Intermediate Updates]
\label{thm:per_agent}
    The per-agent regret of \Cref{alg:async_DOG} with a learning rate $\eta = \sqrt{\frac{\ln |\mathcal{V}_i|}{|\mathcal{V}_i| T (1+\bar{M}_T^i/4)}}$ against an oblivious adversary satisfies
    \begin{equation}
        \frac{\mathbb{E}[\mathrm{Reg}_T]}{T} \leq \Tilde{O}\mleft(\sqrt{\frac{|\mathcal V_i| \bar{M}_T^i}{T}} \mright),
    \end{equation}where $\bar{M}_T^i$ is defined in \cref{eq:Mt_bar_definition}. In the worst case of reward estimates being as far from the truth as possible,  by bounding $\mathbb{E}[M_t^i] \leq |\mathcal{V}_i|\bar{d}_i$, for $\eta = \sqrt{\frac{\ln |\mathcal{V}_i|}{|\mathcal{V}_i| T (1+|\mathcal V_i|\bar{d}_i/4)}}$, it holds true
    \begin{equation}
        \frac{\mathbb{E}[\mathrm{Reg}_T]}{T} \leq \Tilde{O}\mleft(|\mathcal V_i| \sqrt{\frac{\bar{d}_i}{T}} \mright).
    \end{equation}
\end{theorem}
The bound provided by \cite{thune2019nonstochastic} for the \scenario{DEW} algorithm is
\begin{equation}
    \frac{\mathbb{E}[\mathrm{Reg}^{\texttt{DEW}}_T]}{T} \leq \Tilde{O}\mleft(\sqrt{\frac{|\mathcal V_i|+\bar{d}_i}{T}} \mright).
\end{equation}

This means that even in the worst case of \alg's missing action estimates resulting in the worst reward estimates, \alg's regret has an extra $|\mathcal{V}|^2$ factor in front of the delay term. However, we can see how \alg's regret is controlled by the expected maximum cumulative regret $\bar{M}_T^i$. This means having better estimates for neighbors' actions reduces the regret. For example, assume that $|\mathcal{V}_i| = 4$ and that for each round in the window $s \in \{t-\bar{d}+1,\ldots, t\}$ agent $i$'s reward estimates are within $0.25$ of the true reward estimates for all actions on average, i.e., $|\hat{r}_{a,s}(Z^{s}_{k_s}) - \hat{r}_{a,s}(r_{a_i,s})| \leq 0.25$. Then we get an average maximum cumulative error of $\bar{M}_T^i = \bar{d}_i/4$ and the regret term becomes better than \scenario{DEW}'s regret.



\begin{definition}[Centralization of Information~\cite{xu2025communication}]
\label{def:coin}
For each time step $t\in [T]$, consider a function $f_t:2^{\calV_\calN}\mapsto$ $\mathbb{R}$ and a communication network $\{\calN_i\}_{i\myin\calN}$ where each agent $i\in \calN$ has selected an action $a_{i,t}$. Then, at time $t$, agent $i$'s \emph{Centralization Of INformation} is defined as
\begin{equation}\label{eq:ourcurv}
  \ourcurv_{f_t,i} (\calN_{i})\triangleq f_t(a_{i,t}) - f_t(a_{i,t}\,|\,\{a_{j,t}\}_{j\myin\calN_{i}^c}).
\end{equation}
\end{definition}

$\ourcurv_{f_t,i}$ measures how much $a_{i,t}$ can overlap with the actions of agent $i$'s non-neighbors. In the best scenario, where $a_{i,t}$ does not overlap with other actions at all, \ie $f_t(a_{i,t}\,|\,\{a_{j,t}\}_{j\myin\calN_{i}^c})=f_t(a_{i,t})$, then $\ourcurv_{f_t,i}=0$.  In the worst case instead where $a_{i,t}$ is fully redundant, \ie $f_t(a_{i,t}\,|\,\{a_{j,t}\}_{j\myin\calN_{i}^c})=0$, then $\ourcurv_{f_t,i}= f_t(a_{i,t})$.

\begin{definition}[Curvature~\cite{conforti1984submodular}]\label{def:curvature}
        The curvature of a normalized submodular function $f\colon 2^{\calV}\mapsto \mathbb{R}$ is defined as
        \begin{equation}
            \kappa_f\triangleq 1-\min_{\elem\in\calV}{[f(\calV)-f(\calV\setminus\{\elem\})]}/{f(\elem)}.
        \end{equation}
\end{definition}$\kappa_f$ measures how far~$f$ is from modularity: if $\kappa_f=0$, then  $f(\calV)-f(\calV\setminus\{v\})=f(v)$, $\forall v\in\calV$, \ie $f$ is modular. In~contrast, $\kappa_f=1$ in the extreme case where there exist $v\in\calV$ such that $f(\calV)=f(\calV\setminus\{v\})$, \ie $v$~has no contribution in the presence of $\calV\setminus\{v\}$.

\begin{theorem}[\alg's Approximation Performance]
\label{thm:approx_performance}
Over $t\in[T]$, given the communication network $\{\mathcal{N}_i\}_{i\in\mathcal{N}}$,
\alg instructs each agent $i\in\mathcal{N}$ to select actions $\{a_{i,t}\}_{t\in[T]}$
\begin{itemize}[leftmargin=*]
\item If the network is fully centralized, i.e., $\mathcal{N}_i=\mathcal{N}\setminus\{i\}$,
\small\begin{align}
\hspace{-3mm}\mathbb{E}\mleft[\frac{1}{T} \sum_{t=1}^T f_t(\calA_t)\mright]
\ge&
\frac{1}{1+\kappa_f}
\mathbb{E}\mleft[\frac{1}{T} \sum_{t=1}^Tf_t(\calA^{\opt})\mright] \nonumber \\
&-\underbrace{\tilde{\mathcal{O}}\mleft(|\mathcal{N}|\sqrt{\frac{|\bar{\mathcal{V}}|\bar{M}_T}{T}}\mright)}_{\phi(T)}.
\label{eq:thm1_centralized}
\end{align}

\item If the network is fully decentralized, i.e., $\mathcal{N}_i=\emptyset$,
\small\begin{align}
\hspace{-3mm}\mathbb{E}\mleft[\frac{1}{T} \sum_{t=1}^T f_t(\calA_t)\mright]
\ge&
(1-\kappa_f)
\mathbb{E}\mleft[\frac{1}{T} \sum_{t=1}^T f_t(\calA^{\opt})\mright] \nonumber \\
&-\underbrace{\tilde{\mathcal{O}}\mleft(|\mathcal{N}|\sqrt{\frac{|\bar{\mathcal{V}}|\bar{M}_T}{T}}\mright)}_{\chi(T)}.
\label{eq:thm1_decentralized}
\end{align}

\item If the network is anything in between fully centralized and fully decentralized,
i.e., $\mathcal{N}_i\subseteq \mathcal{N}\setminus\{i\}$,
\small\begin{equation}
\begin{aligned}
\mathbb{E}\mleft[\frac{1}{T} \sum_{t=1}^T f_t(\calA_t)\mright]
\ge\;&
\frac{1}{1+\kappa_f}\,
\mathbb{E}\mleft[\frac{1}{T} \sum_{t=1}^T f_t(\calA^{\opt})\mright]\\
&\hspace{-2.5cm}-\frac{\kappa_f}{1+\kappa_f}\sum_{i\in\mathcal{N}}
\mathbb{E}\mleft[\frac{1}{T} \sum_{t=1}^T \ourcurv_{f_t,i}(\mathcal{N}_i)\mright]-\underbrace{\tilde{\mathcal{O}}\mleft(|\mathcal{N}|\sqrt{\frac{|\bar{\mathcal{V}}|\bar{M}_T}{T}}\mright)}_{\psi(T)}.
\end{aligned}
\label{eq:thm1_intermediate}
\end{equation}
\end{itemize}
Particularly, the expectation is due to \alg's internal randomness, and $\tilde{O}(\cdot)$ hides $\log$~terms and $|\bar{\mathcal{V}}| = \max_{i \in \mathcal{N}}|\mathcal{V}_i|$ along with $\bar{M}_T = \max_{i \in \mathcal{N}}\bar{M}_T^i$.
\end{theorem}

As $T \to \infty$, the error terms $\phi(T)$, $\chi(T)$, and $\psi(T)$ in \cref{eq:thm1_centralized,eq:thm1_decentralized,eq:thm1_intermediate} vanish, so the approximation quality of \alg is asymptotically governed by curvature and network structure. The fully connected case achieves the centralized factor ${1}/(1+\kappa_f)$, whereas partial decentralization incurs the additional penalty that depends on $\scenario{coin}_{f_t,i}$, capturing the loss from limited coordination. Thus, larger coordination neighborhoods improve steady-state performance, while $\phi(T),\chi(T),\psi(T)$ only describe transient learning error. Importantly, the $1/(1+\curv_f)$ suboptimality bound with a fully connected network recovers the bound in~\cite{conforti1984submodular} and is near-optimal as the best possible bound for \cref{eq:problem1_obj} is $1-\kappa_f/e$~\cite{sviridenko2017optimal}.\footnote{{The bounds $1/(1+\curv_f)$ and $1-\kappa_f/e$ become $1/2$ and $1-1/e$ when, in the worst case, $\kappa_f=1$.}}%

Finally, we present the convergence analysis of \alg. 
\begin{theorem}[\alg's Convergence Time]\label{thm:convergence}
\alg achieves $\varepsilon$-convergence to near-optimal actions after $\tilde{O}\mleft({|\mathcal{N}|^2 |\bar{\mathcal V}| \bar{M}_T}/{\varepsilon^2}\mright)$ rounds.
\end{theorem}
\paragraph*{Proof} $\tilde{O}\mleft({|\mathcal{N}|^2 |\bar{\mathcal V}| \bar{M}_T}/{\varepsilon^2}\mright)$ rounds are needed to ensure $\phi(T), \chi(T), \psi(T)<\varepsilon$. \qed

\section{Asynchronous Formulation}
\label{sec:async}

We now consider asynchronous agents that run on their own clocks. Particularly, time is indexed by an ideal global (logical) clock $t \in [T]$, but each agent $i$ runs on its own local clock $C_i(\cdot)$, which is a strictly increasing function of physical time, following standard models of distributed systems and clock synchronization \cite{lamport1978time,freris2011fundamental}. Furthermore, let $\tau_i(t)\in\mathbb{R}_{\geq0}$ denote the physical time at which agent $i$ executes the update associated with logical round $t$. Since agents operate on distinct local clocks, the collection $\{\tau_i(t)\}_{i=1}^n$ will never be identical.  To this end, we assume a uniform bound on the resulting timing mismatch between the agents:
\begin{equation}
    |\tau_i(t)-\tau_j(t)| \le \rho,
    \qquad \forall i,j \in \mathcal{N},\ \forall t\ge 1.
    \label{eq:bounded_skew}
\end{equation}This is a reasonable assumption as in distributed systems, local hardware clocks are typically modeled as having bounded drift, while synchronization mechanisms are designed to keep the induced logical-clock skew bounded despite uncertainty in communication latency \cite{lamport1978time,tsitsiklis1986distributed,fan2004gradient,freris2011fundamental}. Also, similar bounded-clock-error assumptions appear in prior decentralized reachability-based control for distributed CPS~\cite{nguyen2023dsc}.


In our asynchronous setting, agent $i$ receives the round-$t$ actions of its neighbors at physical times $\tau_j(t)+\delta^j_{i}(t)$
where $\tau_j(t)$ is the physical time at which agent $j \in \mathcal{N}_i$ executes its round-$t$ action and $\delta^j_{i}(t)$ is the communication delay for the round $t$ information transmitted from agent $j$ to agent $i$.


Since all agents are running their own clocks, for each  $t$, every agent will likely execute its action at a different time. Thus, the global objective as in \cref{eq:problem1_obj} will lose its meaning. To this end, we define a new version of the time-varying submodular function and a corresponding global objective that accurately represents the asynchronous setting.

\subsection{Asynchronous reward definition and assumptions}

Because the environment evolves in continuous time while agents execute at different
physical instants, the reward of a round depends not only on \emph{which} actions were
taken but on \emph{when} each took effect. This motivates the following two definitions for establishing the reward for the asynchronous execution regime.

\begin{definition}[Action-instant ground set]
\label{def:groundset}
The \emph{action-instant} ground set is
\begin{equation}
\label{eq:groundset}
\mathcal{E} \;\triangleq\; \mathcal{V}\times\mathbb{R}_{\ge0},
\end{equation}
whose elements $e=(a,\tau)$ denote ``action $a$ executed at physical time $\tau$''. A
\emph{deployment schedule} is a finite subset $\mathcal{D}\subseteq\mathcal{E}$.
\end{definition}
 
Here $\mathcal{V}=\bigsqcup_{i\in\mathcal{N}}\mathcal{V}_i$ is a disjoint union, so
distinct agents contribute distinct elements even when
they select the same action at the same instant.
 
\begin{definition}[Time-stamped reward function]
\label{def:reward}
The \emph{time-stamped reward function} is a set function on the action-instant ground set,
\begin{equation}
\label{eq:reward}
F : 2^{\mathcal{E}} \longrightarrow \mathbb{R}_{\ge0},
\end{equation}
assigning to each deployment schedule $\mathcal{D}$ the reward $F(\mathcal{D})$ incurred by
executing its actions at their associated instants. Its marginals are written
$F(e\mid S)\triangleq F(S\cup\{e\})-F(S)$ for $e\in\mathcal{E}$ and $S\subseteq\mathcal{E}$.
\end{definition}
 
One feature of Definition~\ref{def:reward} worth stating explicitly is that $F$ is an ordinary set function, so no
ordering of agents by execution time is required to define the round reward, which is a point we
return to in Remark~\ref{rem:ordering}.

To make this reward consistent with the synchronous setting, and to allow for regret analysis, we also have the following submodularity and time-lipschitzness conditions.
 
For a fixed $\tau\ge0$ we write the \emph{synchronous restriction}
\begin{equation}
F|_\tau(A) \;\triangleq\; F\big(\{(a,\tau) : a\in A\}\big),
\qquad A\subseteq\mathcal{V},
\end{equation}
to denote the synchronous version of the reward that comes from the corresponding actions being executed at some common time $\tau$.
 
\begin{assumption}[Submodularity]
\label{asm:A}
For each fixed $\tau\ge0$, the synchronous restriction $F|_\tau$ is normalized, non-decreasing,
submodular, and 2nd-order submodular on $2^{\mathcal{V}}$ in the sense of
\Cref{def:submodular} and \cref{eq:second_order_submod}.
\end{assumption}
 
\begin{assumption}[Boundedness]
\label{asm:B}
$F$ is normalized and non-decreasing on $\mathcal{E}$, and there is a known finite
$R_{\max}>0$ with $F(\mathcal{D})\le R_{\max}$ for every deployment schedule
$\mathcal{D}\subseteq\mathcal{E}$ with $|\mathcal{D}|\le|\mathcal{N}|$.
\end{assumption}

Monotonicity and normalization place every marginal in $[0,R_{\max}]$, so after rescaling
$F$ by $1/R_{\max}$ --- which we assume throughout without further comment --- the per-agent
rewards lie in $[0,1]$ as Theorem~\ref{thm:per_agent} requires. Requiring
$R_{\max}$ to be known rather than merely finite is what allows Algorithm~\ref{alg:async_DOG} to
form its importance-weighted estimates online. This is the boundedness assumption already
adopted in Section~\ref{sec:problem}, where such a conservative bound is taken to follow from knowledge of worst-case environment dynamics.
 
To quantify the effect of clock mismatch we bound how much the reward can oscillate as
execution times move. Write $S\sim_\rho S'$ if $S,S'\subseteq\mathcal{E}$ admit an action-preserving
bijection whose paired timestamps differ by at most $\rho$.
 
\begin{assumption}[Bounded temporal oscillation]
\label{asm:C}
The quantities
\begin{align}
\label{eq:osc1}
\mathrm{osc}_1(\rho) &\;\triangleq\; \sup_{\substack{a\in\mathcal{V},\ S\subseteq\mathcal{E}\\ |\tau-\tau'|\le\rho}}
\big|F\big((a,\tau)\mid S\big)-F\big((a,\tau')\mid S\big)\big|,\\[2pt]
\label{eq:osc2}
\mathrm{osc}_2(\rho;m) &\;\triangleq\; \sup_{\substack{e\in\mathcal{E},\ S\sim_\rho S'\\ |S|=m}}
\big|F(e\mid S)-F(e\mid S')\big|,
\end{align}
are finite for every $\rho\ge0$ and every integer $m\ge0$.
\end{assumption}
 
$\mathrm{osc}_1$ bounds the change in one action's marginal contribution when only that
action's execution time is perturbed. $\mathrm{osc}_2$ bounds the change when the entire
surrounding context is re-timestamped in a single move, holding the action fixed. The second argument in $\mathrm{osc}_2(\rho;m)$ records the size of the context being re-timestamped. Where a single
network-wide constant is needed we abbreviate
$\mathrm{osc}_2(\rho)\triangleq\max_{i\in\mathcal{N}}\mathrm{osc}_2(\rho;|\mathcal{N}_i|)$.

Since both $\mathrm{osc}_1$ and $\mathrm{osc}_2$ are supremum over sets that increase in size with $\rho$, both oscillation bounds are non-decreasing in $\rho$, which implies that the
limits
\begin{equation}
\label{eq:floor}
c_k \;\triangleq\; \mathrm{osc}_k(0^+) \;=\; \inf_{\rho>0}\mathrm{osc}_k(\rho),
\qquad k=1,2,
\end{equation}
exist without further assumption. Assumption~\ref{asm:C} asks only for
finiteness and the values $c_1,c_2$ are based on the specific $F$ being considered. Thus, $c_1,c_2$ classify the reward as described in Remark~\ref{rem:dichotomy}.
 
\begin{remark}[Why an oscillation bound rather than a Lipschitz constant]
\label{rem:whyosc}
Assumption~\ref{asm:C} replaces the evaluation-time and deployment-time Lipschitz
conditions used in the preliminary version of this work. A Lipschitz condition is the
special case $\mathrm{osc}(\rho)=L\rho$, and it is \emph{not} available in general: for a
coverage reward the value jumps by a full marginal when a target crosses a
field-of-view boundary, so no finite $L$ exists. Assumption~\ref{asm:C} asks only that the
oscillation be finite, which such a reward satisfies --- a single crossing changes the
reward by one target --- and it therefore covers rewards that admit no Lipschitz constant
at all. What is lost in that case is not the bound but its behavior as $\rho\to0$ (see
\Cref{rem:dichotomy} and \Cref{sec:sim-info}).
\end{remark}
 
\begin{remark}[Two regimes]
\label{rem:dichotomy}
The floors $c_1, c_2$ of \cref{eq:floor} classify $F$ into two regimes, and every result below holds in both. If $c_1 = c_2 = 0$, the reward is continuous in the execution times and the asynchrony penalty of \Cref{thm:async} vanishes as the clock mismatch $\rho \to 0$, so tighter synchronization buys proportionate performance at the rate at which $\mathrm{osc}_k$ vanishes. If instead $c_1 > 0$ or $c_2 > 0$, the reward can jump when a discontinuity falls inside the execution window. Synchronizing the agents exactly still removes the penalty, since $\mathrm{osc}_k(0) = 0$ for every $F$, but the approach to that limit is abrupt: once $\rho$ is small enough that $\mathrm{osc}_k$ has flattened, tightening the clocks further buys nothing, and the average regret retains a term that does not vanish with the horizon $T$ by \Cref{lem:conversion}. This floor is a worst-case quantity, though, and is reached only when a discontinuity actually falls inside the window, so in practice the penalty may be far smaller.
\end{remark}

 
\subsection{Asynchronous regret analysis}
 
Fix a round $t$ and a reference instant $\tau_{\mathrm{ref}}(t)$ satisfying
$|\tau_i(t)-\tau_{\mathrm{ref}}(t)|\le\rho$ for all $i\in\mathcal{N}$. Define
\begin{align}
\mathcal{D}_t &\triangleq \{(a_{i,t},\tau_i(t))\}_{i\in\mathcal{N}},
&&\text{(realized schedule)}\\
\bar{\mathcal{D}}_t &\triangleq \{(a_{i,t},\tau_{\mathrm{ref}}(t))\}_{i\in\mathcal{N}},
&&\text{(synchronous counterfactual)}
\end{align}
so that $F(\mathcal{D}_t)$ is the \emph{asynchronous global reward} for round $t$ and
$F(\bar{\mathcal{D}}_t)=f_t(\{a_{i,t}\}_{i\in\mathcal{N}})$ is the synchronous reward of
the same actions. For agent $i$ and any candidate action $a\in\mathcal{V}_i$ we write
\begin{align}
r^{\mathrm{async}}_{i,t}(a) &\triangleq F\big((a,\tau_i(t))\mid \{(a_{j,t},\tau_j(t))\}_{j\in\mathcal{N}_i}\big),\\
r^{\mathrm{sync}}_{i,t}(a)  &\triangleq f_t\big(a \mid \{a_{j,t}\}_{j\in\mathcal{N}_i}\big).
\end{align}
Agent $i$ observes $r^{\mathrm{async}}_{i,t}(a_{i,t})$, whereas the static regret in
\cref{eq:static_regret_agent_i} is stated in terms of $r^{\mathrm{sync}}_{i,t}$. We therefore
carry both, writing for $\bullet\in\{\mathrm{sync},\mathrm{async}\}$
\begin{equation}
\label{eq:regsyncasync}
\mathrm{Reg}^{\bullet}_T \;\triangleq\;
\max_{a\in\mathcal{V}_i}\sum_{t=1}^{T} r^{\bullet}_{i,t}(a)
\;-\;\sum_{t=1}^{T} r^{\bullet}_{i,t}(a_{i,t}),
\end{equation}
suppressing the agent index. Both compare the played sequence $\{a_{i,t}\}_{t\in[T]}$
against the best fixed action in hindsight, and differ only in which reward function scores
it: $\mathrm{Reg}^{\mathrm{sync}}_T$ is exactly the static regret of
\cref{eq:static_regret_agent_i}, while $\mathrm{Reg}^{\mathrm{async}}_T$ is the
regret with respect to the rewards the agent actually observes, which is what
Theorem~\ref{thm:per_agent} bounds. In the synchronous setting $\tau_i(t)=\tau_{\mathrm{ref}}$
for all $i$ and the two coincide.
 
\begin{remark}[No execution-time ordering is required]
\label{rem:ordering}
$F(\mathcal{D}_t)$ is a set function evaluated on a set, so the round reward is
independent of any ordering of the agents. It may still be decomposed into per-agent
marginals along \emph{any} permutation $\pi$,
\begin{equation}
F(\mathcal{D}_t)=\sum_{k=1}^{|\mathcal{N}|}
F\big(e_{\pi(k)}\mid\{e_{\pi(1)},\dots,e_{\pi(k-1)}\}\big),
\end{equation}
which telescopes for every $\pi$. Ordering by execution time in proofs for the following lemmas and theorems is not part of the definition, and near-simultaneous executions or arbitrary
tie-breaking do not affect the round reward.
\end{remark}

\begin{lemma}[Global gap]
\label{lem:global}
Under Assumption~\ref{asm:C}, for every round $t$,
\begin{equation}
\big|F(\mathcal{D}_t)-F(\bar{\mathcal{D}}_t)\big| \;\le\; |\mathcal{N}|\,\mathrm{osc}_1(\rho).
\end{equation}
\end{lemma}
 
\begin{lemma}[Marginal gap]
\label{lem:marginal}
Under Assumption~\ref{asm:C}, for every agent $i$ and round $t$,
\begin{equation}
\gamma_{i,t} \;\triangleq\; \max_{a\in\mathcal{V}_i}\big|r^{\mathrm{async}}_{i,t}(a)-r^{\mathrm{sync}}_{i,t}(a)\big|
\;\le\; \mathrm{osc}_1(\rho)+\mathrm{osc}_2\big(\rho;|\mathcal{N}_i|\big).
\end{equation}
\end{lemma}
 
\Cref{lem:global} bounds the gap between the asynchronous and synchronous global rewards for a single round, with the played actions held fixed, whereas \Cref{lem:marginal} bounds the corresponding gap in the marginal reward local to each agent. The two differ in more than scope. The maximum in \Cref{lem:marginal} runs over all candidate actions in $\mathcal{V}_i$, including those that the agent never plays and whose rewards it therefore never observes. This is what \Cref{lem:conversion} requires, since that argument is evaluated at the best action in hindsight, and a bound covering only played actions cannot support it.
 
 
 
\begin{lemma}[Regret conversion]
\label{lem:conversion}
Fix agent $i$, let $\gamma\triangleq\max_{t\in[T]}\gamma_{i,t}$, and let
$\mathrm{Reg}^{\mathrm{sync}}_T$ and $\mathrm{Reg}^{\mathrm{async}}_T$ be as
in \cref{eq:regsyncasync}. Then,
\begin{equation}
\mathrm{Reg}^{\mathrm{sync}}_T \;\le\; \mathrm{Reg}^{\mathrm{async}}_T + 2\sum_{t=1}^{T}\gamma_{i,t}
\;\le\; \mathrm{Reg}^{\mathrm{async}}_T + 2T\gamma .
\end{equation}
\end{lemma}
 
Lemma~\ref{lem:conversion} connects the two halves of the analysis by converting the
regret Theorem~\ref{thm:per_agent} bounds into the regret Theorem~\ref{thm:approx_performance} consumes.
Note that it compares a single played sequence under two reward functions rather than comparing two independent synchronous and asynchronous runs. The bound holds for every realization of timestamps, actions, and algorithmic randomness, so it passes through
expectations without further argument.
 
\begin{theorem}[Approximation performance under asynchrony]
\label{thm:async}
Let $\widetilde{\mathrm{osc}}(\rho)\triangleq\max\{\mathrm{osc}_1(\rho),\mathrm{osc}_2(\rho)\}$.
Under \Cref{asm:A,asm:B,asm:C}, each bound in
\cref{eq:thm1_centralized,eq:thm1_decentralized,eq:thm1_intermediate} remains valid with
$f_t(\mathcal{A}_t)$ replaced by $F(\mathcal{D}_t)$ and with an additional loss
\begin{equation}
\label{eq:Gamma}
\Gamma(\rho)\le\frac{2|\mathcal{N}|}{1+\kappa_f}\big(\mathrm{osc}_1+\mathrm{osc}_2\big)(\rho)
+|\mathcal{N}|\,\mathrm{osc}_1(\rho)
\le 5\,|\mathcal{N}|\,\widetilde{\mathrm{osc}}(\rho)
\end{equation}
subtracted from the right-hand side. In particular
$\Gamma(\rho)=O\big(|\mathcal{N}|\,\widetilde{\mathrm{osc}}(\rho)\big)$, with
$c_1,c_2$ as in \cref{eq:floor},
\begin{equation}
\lim_{\rho\to0}\Gamma(\rho)\;\le\;\frac{2|\mathcal{N}|}{1+\kappa_f}(c_1+c_2)+|\mathcal{N}|\,c_1 .
\end{equation}
\end{theorem}
 
The penalty is linear in the number of agents and carries no explicit topology factor, but
it is not topology-independent. $\mathrm{osc}_2$ is evaluated at each agent's neighborhood
size, and \Cref{asm:A,asm:B,asm:C} do not determine how it varies with that size. We can show $\mathrm{osc}_2$ to be bounded as $\mathrm{osc}_2(\rho;m)\le 2m\,\mathrm{osc}_1(\rho)$. However, this is an upper bound only and
whether $\mathrm{osc}_2$ in fact grows, saturates, or decreases with $m$ is a property of
the reward function specific to the problem being considered.

\section{Simulations}
\label{sec:simulations}

We evaluate \alg in two settings. \Cref{sec:sim-coverage} reports the multi-camera coverage
task, which isolates the benefit of intermediate updates under increasing communication
delay. \Cref{sec:sim-info} reports a second, independent camera coverage simulation with a continuous-valued
information-gathering reward, built to test the asynchronous theory of
\Cref{sec:async} directly.

\subsection{Multi-camera coverage under communication delays}
\label{sec:sim-coverage}

\begin{figure}[t]
    \centering
    \includegraphics[width=0.85\columnwidth]{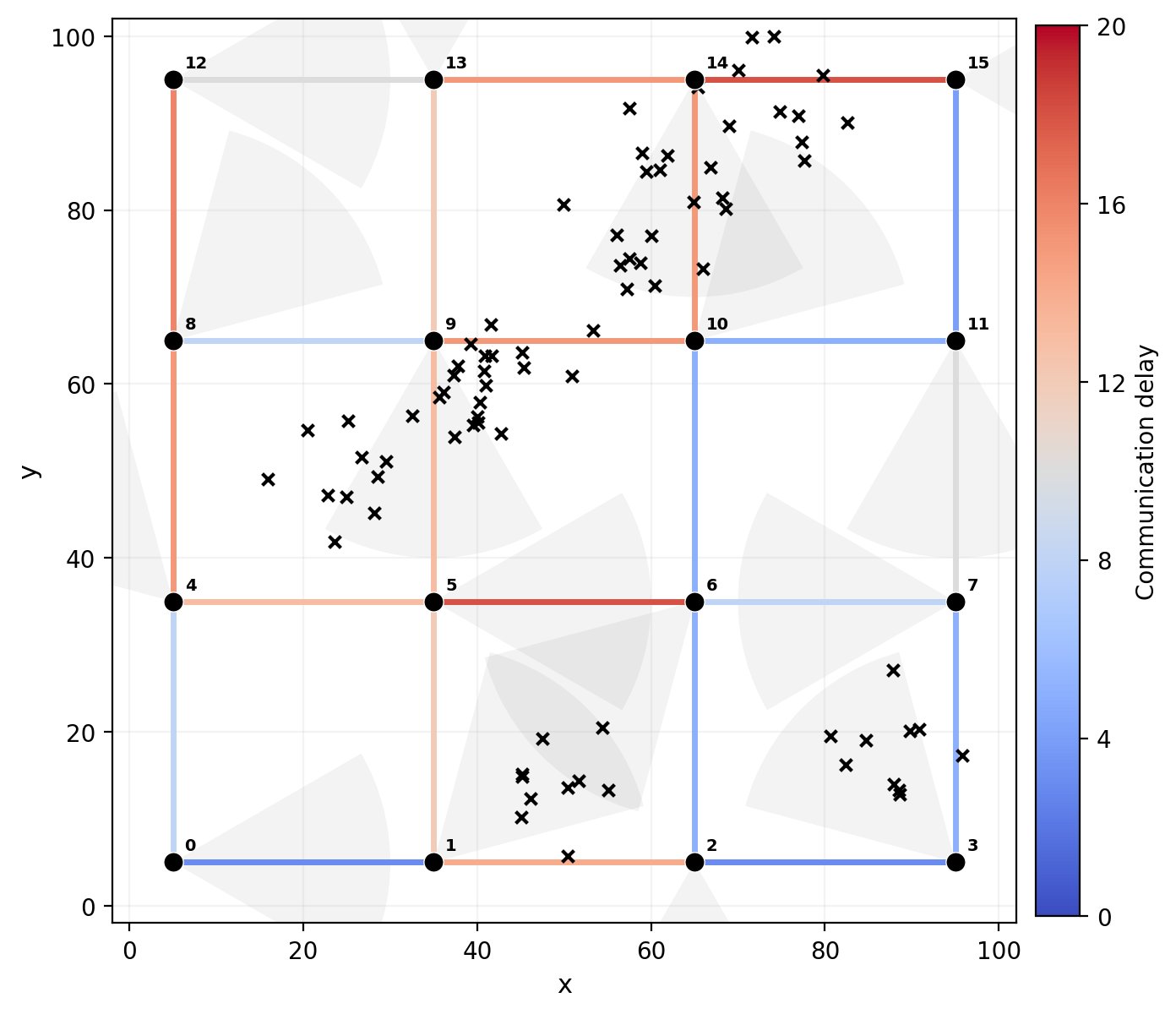}
    \vspace{-3.5mm}
    \caption{Simulation layout and sample snapshot of camera (black vertices) and target (black crosses) configuration. $16$ cameras are placed on a $4 \times 4$ grid over a $100 \times 100$ workspace. Each camera has a sector FOV with half-angle $30^\circ$ and sensing range of $20$ units (light gray wedges show the selected heading of each camera). Colored edges denote the one-hop communication links between grid neighbors, with warmer colors representing higher delays.}
    \label{fig:layout}
\end{figure}

\begin{figure*}[t]
    \centering
    \setlength{\tabcolsep}{2pt}

    \begin{subfigure}[t]{0.329\textwidth}
        \centering
        \includegraphics[width=\linewidth]{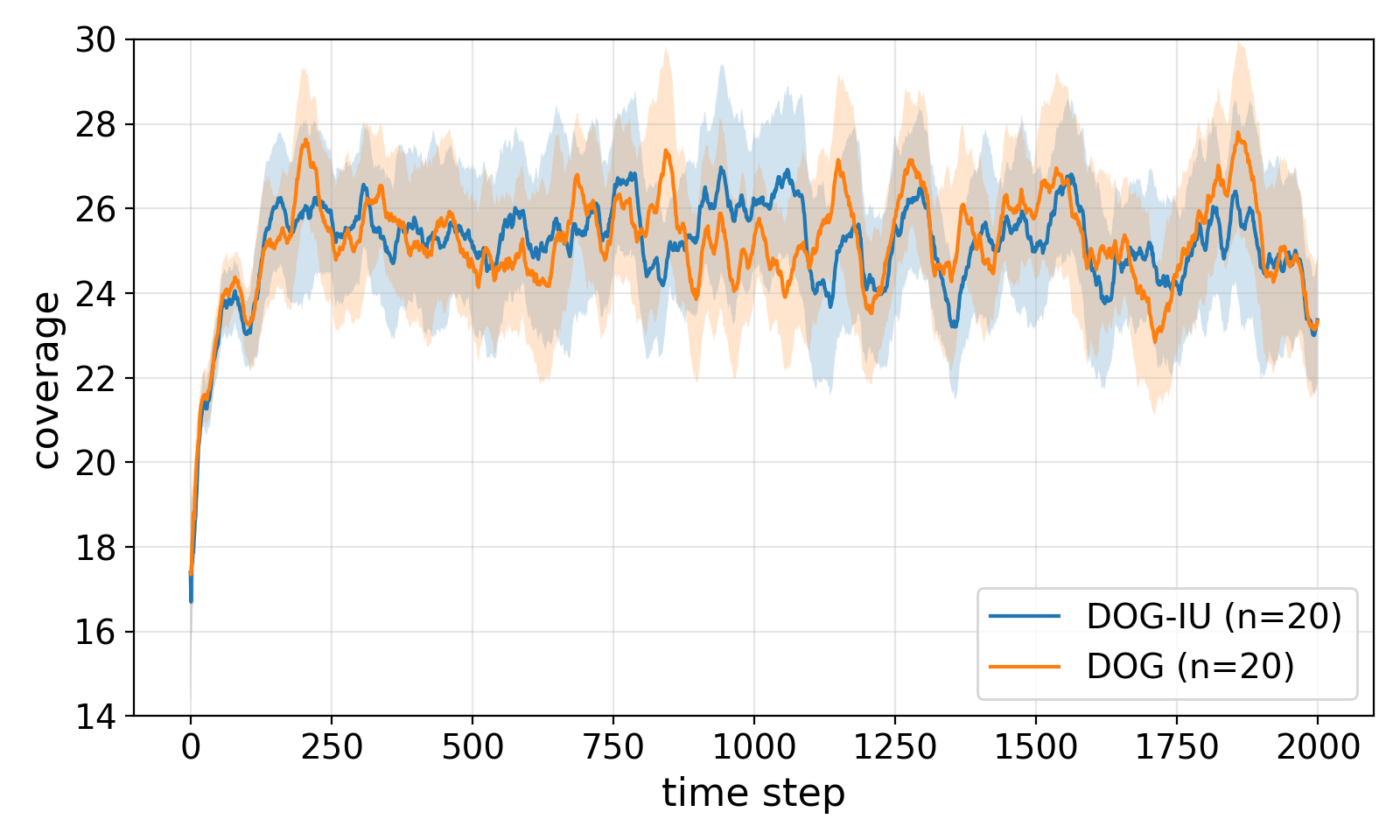}
        \vspace{-7mm}
        \caption{$\bar{d}=1$}
        \label{fig:cov_d1}
    \end{subfigure}
    \hfill
    \begin{subfigure}[t]{0.329\textwidth}
        \centering
        \includegraphics[width=\linewidth]{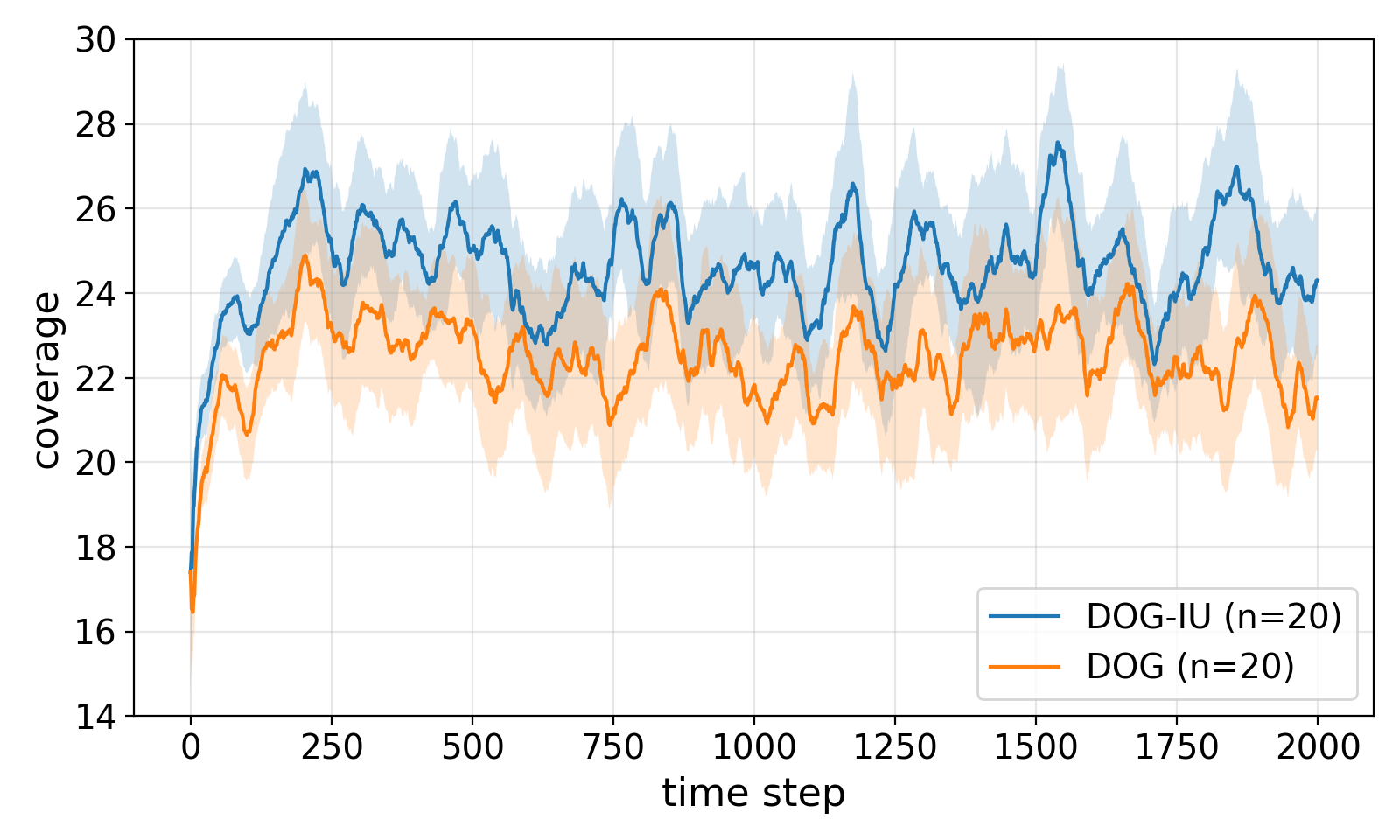}
        \vspace{-7mm}
        \caption{$\bar{d}=5$}
        \label{fig:cov_d5}
    \end{subfigure}
    \hfill
    \begin{subfigure}[t]{0.329\textwidth}
        \centering
        \includegraphics[width=\linewidth]{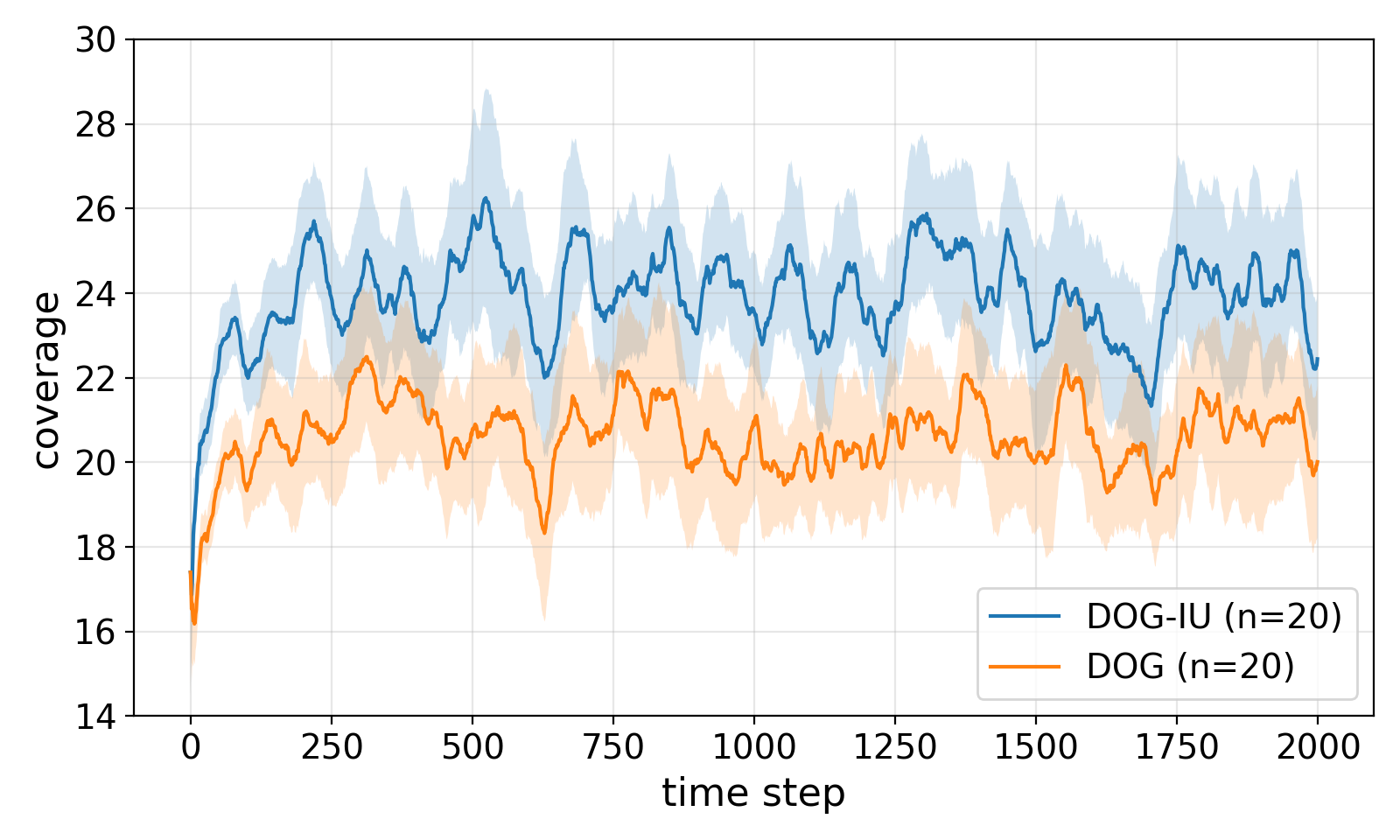}
        \vspace{-7mm}
        \caption{$\bar{d}=10$}
        \label{fig:cov_d10}
    \end{subfigure}

    \begin{subfigure}[t]{0.329\textwidth}
        \centering
        \includegraphics[width=\linewidth]{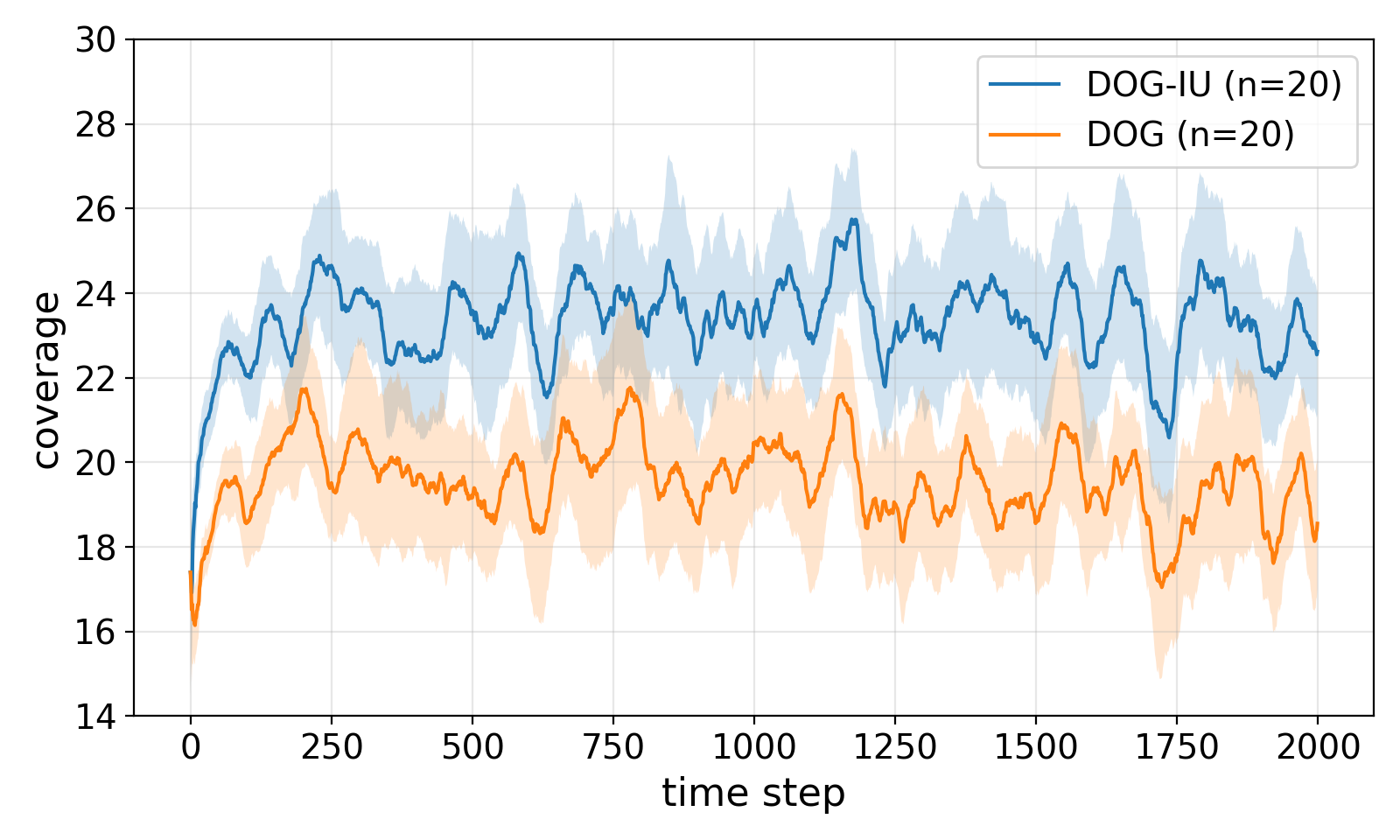}
        \vspace{-7mm}
        \caption{$\bar{d}=15$}
        \label{fig:cov_d15}
    \end{subfigure}
    \hfill
    \begin{subfigure}[t]{0.329\textwidth}
        \centering
        \includegraphics[width=\linewidth]{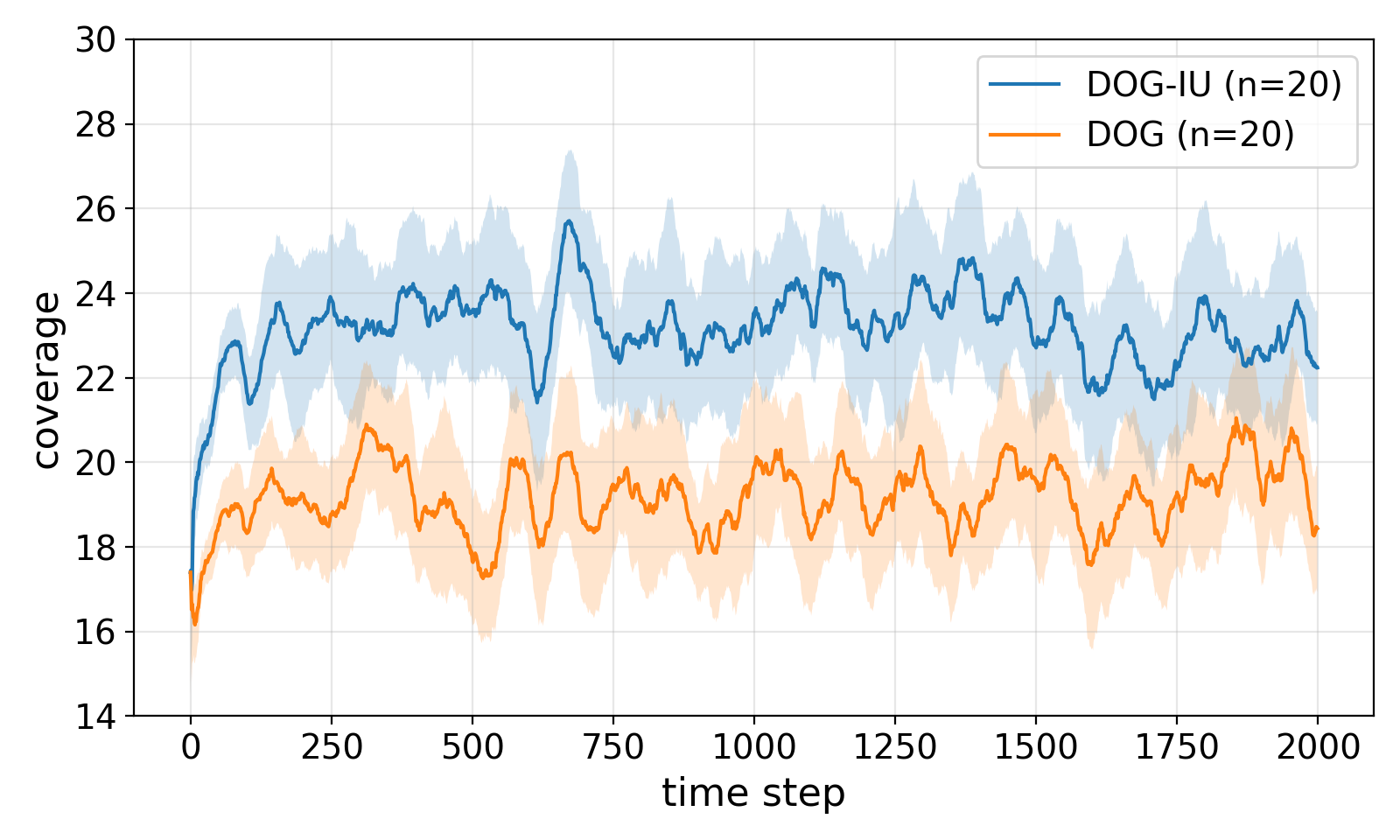}
        \vspace{-7mm}
        \caption{$\bar{d}=20$}
        \label{fig:cov_d20}
    \end{subfigure}
    \hfill
    \begin{subfigure}[t]{0.329\textwidth}
        \centering
        \includegraphics[width=\linewidth]{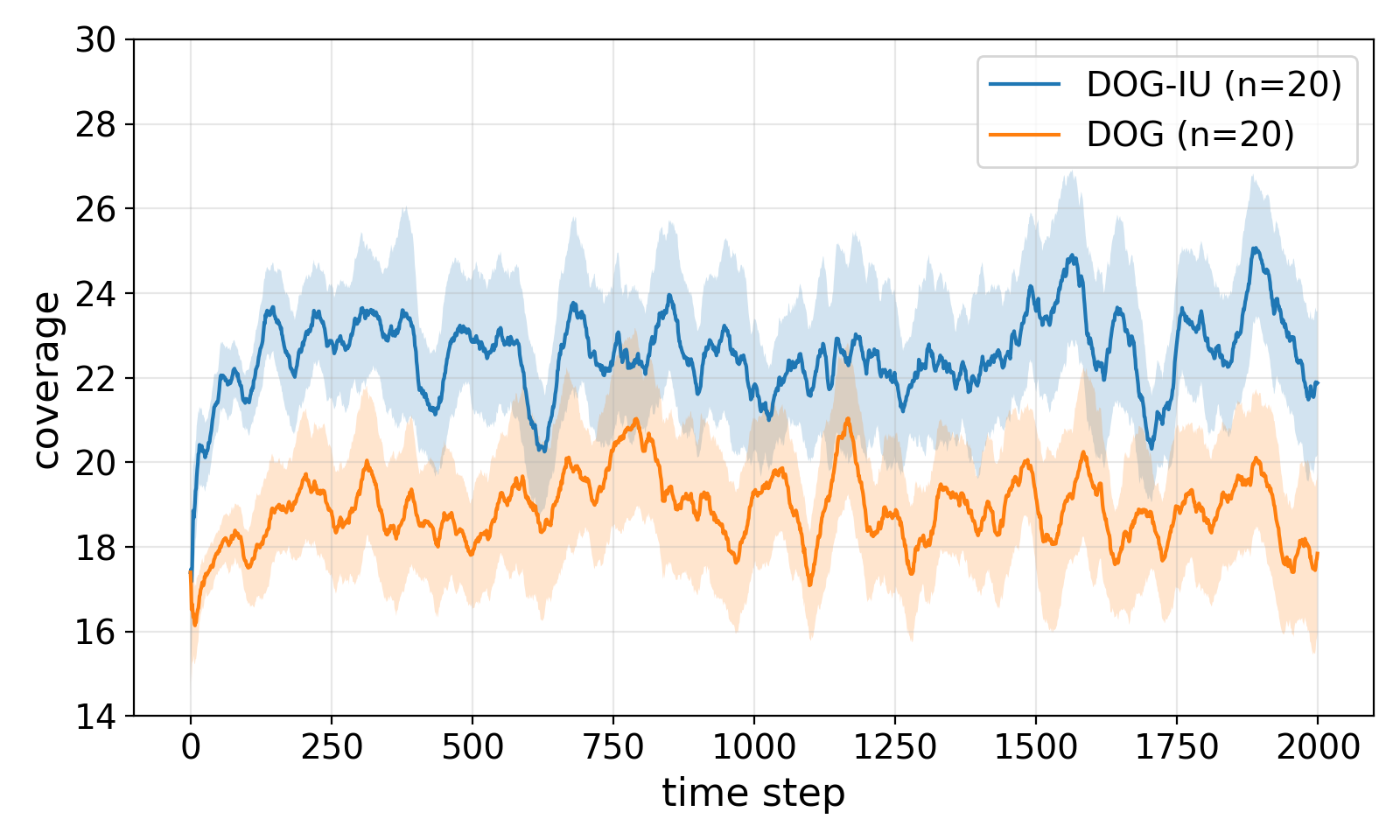}
        \vspace{-7mm}
        \caption{$\bar{d}=30$}
        \label{fig:cov_d30}
    \end{subfigure}
    \caption{Coverage over time (mean $\pm 95\%$ CI, $n=20$ runs, running average over $50$ time steps) for \alg and \scenario{DOG} under increasing maximum one-hop delay $\bar{d}$ on a $100\times 100$ workspace with $16$ grid-placed cameras, $8$ headings, and $80$ clustered targets (8 clusters). The gap widens as delay increases, illustrating the benefit of intermediate updates under more severe communication delays.}
    \label{fig:coverage}
    \vspace{-1.2em}
\end{figure*}

We evaluate \alg in the asynchronous setting, against the baseline \dog in the synchronous setting, under increasing communication delays, on a target-monitoring task. \dog applies the same \scenario{EXP3}-style update as \alg but defers all weight updates for round $t$ until the actions of \emph{all neighbors} for that round have been received.

\textbf{Setup.}
We consider $|\mathcal{N}| = 16$ cameras placed on a $100 \times 100$ workspace as shown in \Cref{fig:layout}. Each camera selects one of $|\calV_i| = 8$ discrete headings per round. The communication graph connects each agent to its immediate grid neighbors (colored edges in \Cref{fig:layout}). We restrict the cameras to one-hop communication.

\textbf{Targets.}
To induce a non-stationary coverage landscape, $80$ targets are organized into $8$ clusters. Each cluster shares a velocity vector of magnitude $1.0$ units/step whose heading is resampled every $30$ steps; individual targets receive i.i.d.\ Gaussian noise ($\sigma = 0.005$) and are reflected at boundaries.

\textbf{Delays and Learning Rate.}
Communication delays are sampled from a uniform distribution for each round, i.e., delays are i.i.d.\ $d^j_{i,t} \sim \mathrm{Unif}\{0, \ldots, \bar{d}\}$; \eg for $\bar{d} = 10$, the average delay for any given communication link will be 5 rounds. Both algorithms use a learning rate of $\eta_i = c\sqrt{\ln|V_i|/((|V_i| + \bar{d})T)}$. Since $\bar{M}_T$ cannot generally be known a priori, we use the learning rate of \scenario{DEW}/\dog, which is of a similar order. Additionally, a scaling factor of $c=14$ amplifies the per-update weight shift, benefiting \alg because its estimation-correction scheme provides more opportunities to react to changes in the environment. This scaling factor is required to make both algorithms adapt to the fast-changing environment.

\textbf{Asynchrony.}
We run \alg in asynchronous mode with a timing mismatch bound of $\rho = 0.3$, that is, the difference between the measurement/action execution times of any two agents will be within $0.3$ of the round duration. In terms of the simulation, we run a global clock $T_{global}$ and for each agent uniformly sample $\tau_i \sim (t_{global}-0.15,t_{global}+0.15)$ along with $\tau_{\mathrm{ref}}(t)=t_{global}$, so that \cref{eq:bounded_skew} holds by construction.

\textbf{Action Estimation.}
Each agent estimates the missing action for a neighbor as that neighbor's last known action.

\textbf{Results.}
Each configuration is evaluated over $n = 20$ Monte Carlo runs with $T = 2000$, using the same environment realization (random seed) for both algorithms. \Cref{fig:coverage} reports coverage trajectories (mean $\pm\, 95\%$ CI). For $\bar{d} = 1$ (\Cref{fig:coverage}a), \alg and \dog perform identically, confirming that even at small delays \alg matches \dog. As the delays increase to $\bar{d} = 5$ and beyond (\Cref{fig:coverage}b-f), a gap emerges between the two algorithms. At $\bar{d} = 10$, \dog defers each round's update by up to $10$ steps, during which the target cluster locations can change substantially ($10$ units of displacement). \alg begins updating immediately using reward estimates conditioned on the full neighborhood and corrects as true actions arrive. At $\bar{d} = 20$ and $\bar{d} = 30$ (\Cref{fig:coverage}e-f), we see that \dog is effectively not able to learn, while \alg is still able to maintain a performance gap of 3-5 targets, which is a roughly $20\%$ advantage. \dog's updates lag by up to $1\%$ of the horizon, while \alg's early estimates, despite being potentially incorrect for some rounds, steer the policy toward better actions before the environment shifts.

The asynchrony penalty of \Cref{thm:async} is not pronounced here. The coverage reward is integer-valued and changes only when a target crosses a field-of-view boundary, so it lies in the second regime of \Cref{rem:dichotomy}: its oscillation bounds are finite but do not vanish at the origin. The resulting floor is a worst-case quantity, attained only when such a crossing falls inside the execution window, and with slowly moving targets those configurations are rare. Isolating and measuring the asynchrony penalty therefore calls for a reward that is continuous in the execution times, which motivates the second simulation.

\subsection{Information gathering under asynchronous execution}
\label{sec:sim-info}

\begin{figure}[t]
    \centering
    \includegraphics[width=\columnwidth]{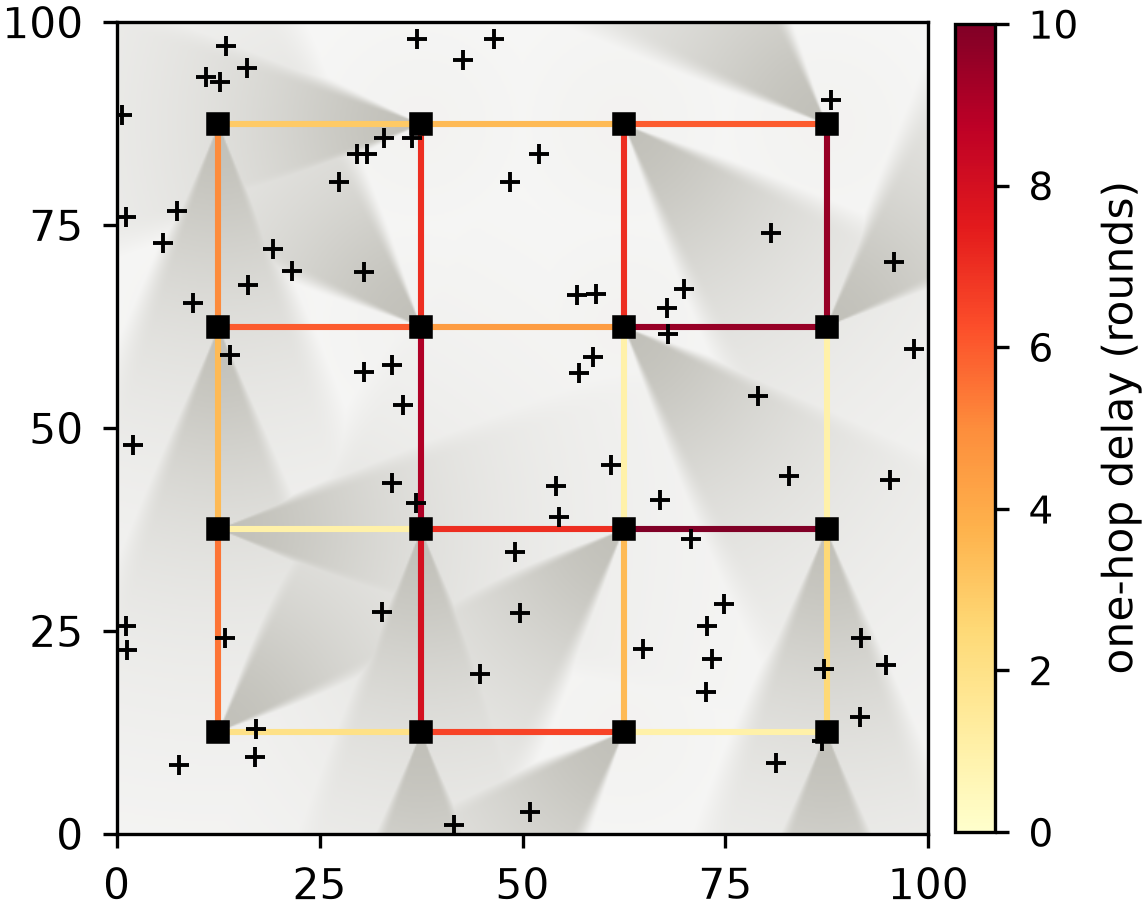}
    \vspace{-3.5mm}
    \caption{Layout and sample snapshot of the information-gathering simulation. $16$ agents (black squares) are fixed at the vertices of a $4\times4$ grid over a $100\times100$ workspace, each orienting a $45^\circ$ sensing sector. Gray shading indicates sensing accuracy: it tapers smoothly across the sector edge to the floor $g_{\mathrm{out}}$ and decays with range, rather than terminating at a hard field-of-view boundary, which is what makes the reward continuous in the execution times. $72$ targets (black crosses) in $8$ clusters follow linear-Gaussian dynamics with continuous-time maneuvers. Colored edges denote the one-hop communication links, with warmer colors representing higher delays.}
    \label{fig:info_layout}
\end{figure}

\begin{figure}[t]
    \centering
    \begin{subfigure}[t]{\columnwidth}
        \centering
        \includegraphics[width=\linewidth]{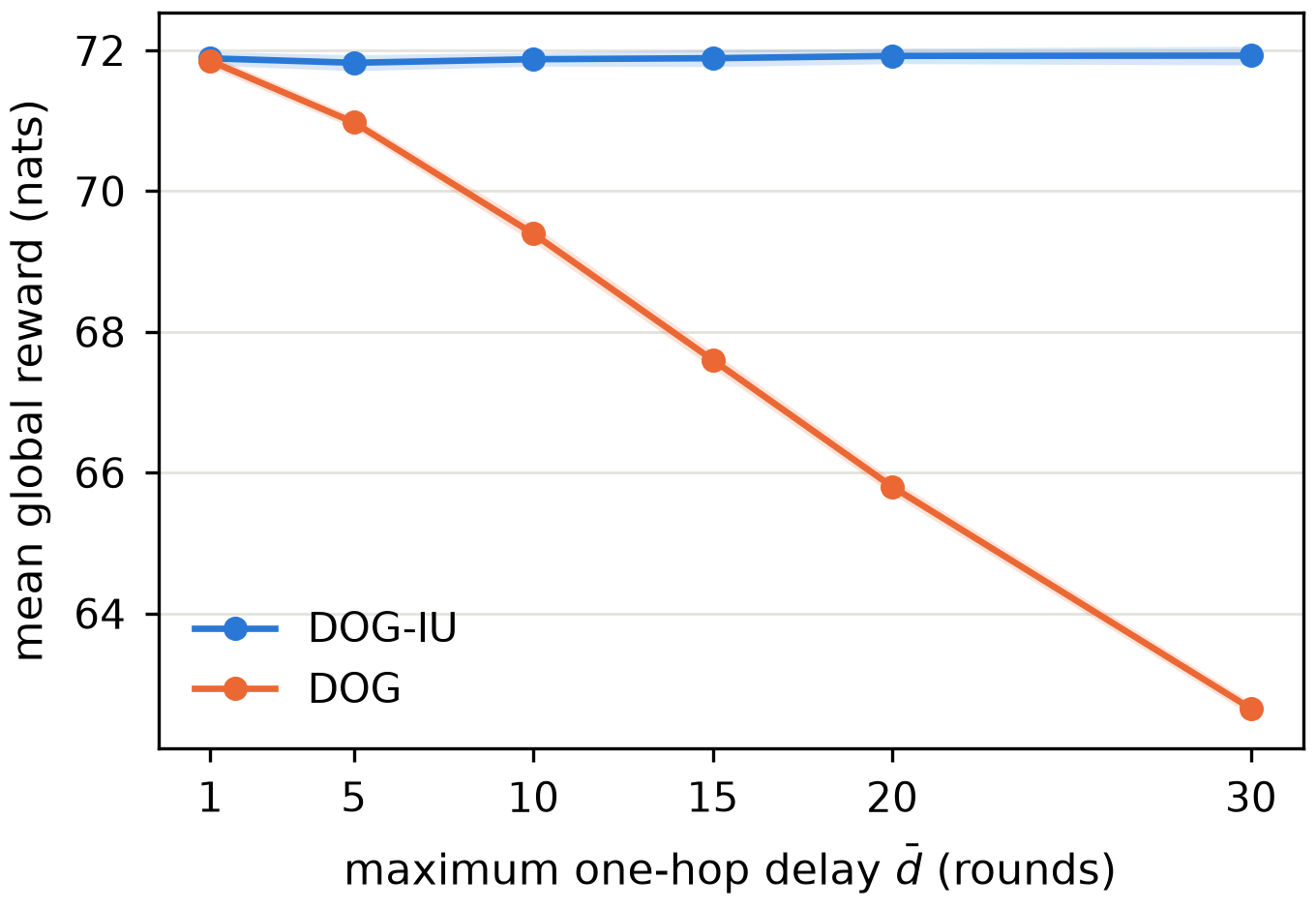}
        \vspace{-7mm}
        \caption{}
        \label{fig:info_delay}
    \end{subfigure}

    \vspace{1mm}

    \begin{subfigure}[t]{\columnwidth}
        \centering
        \includegraphics[width=\linewidth]{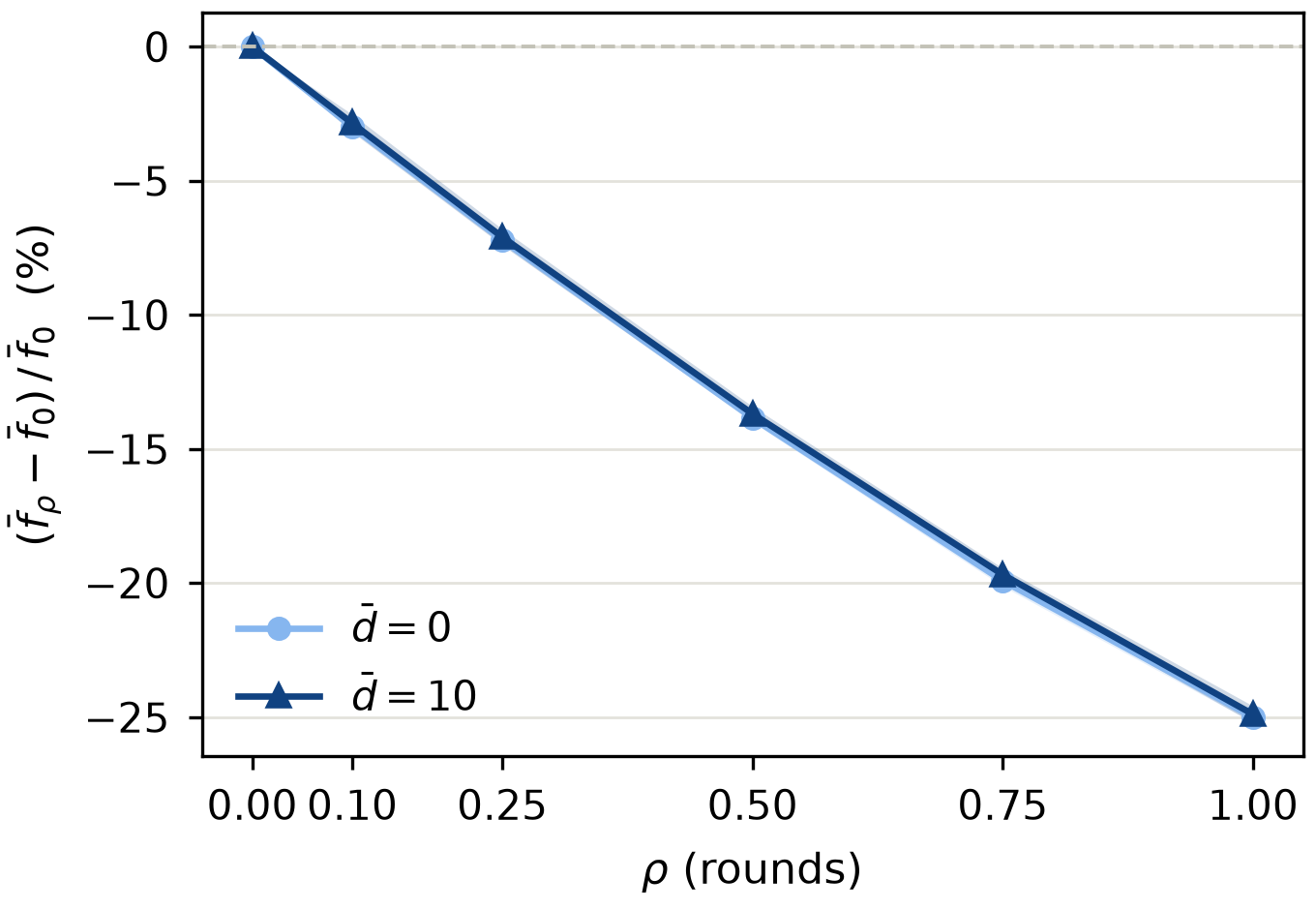}
        \vspace{-7mm}
        \caption{}
        \label{fig:info_rho}
    \end{subfigure}
    \caption{Preliminary results on the information-gathering task; shaded bands are $95\%$ confidence intervals and are narrower than the plot markers throughout. (a) Mean reward against the maximum one-hop delay $\bar{d}$ ($n=20$ seed-paired runs): \alg is essentially unaffected by delay, whereas \dog degrades monotonically, so the advantage of intermediate updates grows with $\bar{d}$ and reaches $14.81\%$ at $\bar{d}=30$. (b) Reward relative to synchronous execution against the timing mismatch bound $\rho$ ($n=10$ seed-paired runs): degradation is smooth and close to proportional in $\rho$, and the curves at $\bar{d}=0$ and $\bar{d}=10$ coincide, indicating that communication delay and clock mismatch act as separate cost channels.}
    \label{fig:info_results}
\end{figure}

The coverage reward of \Cref{sec:sim-coverage} admits no oscillation bound that vanishes at the origin, and its submodularity is asserted rather than checked. We therefore built a second simulation whose reward is continuous in the execution times and whose structural assumptions are verifiable on realized states.

\textbf{Setup.}
$|\mathcal{N}|=16$ agents are fixed at the vertices of a $4\times4$ grid over a $100\times100$ workspace with the same one-hop communication graph as before (\Cref{fig:info_layout}). Each agent orients a $45^\circ$ sensing sector, so $|\calV_i|=8$; sensing accuracy varies with bearing through a $C^1$ raised-cosine gain $g(\cdot)\in[g_{\mathrm{out}},1]$ with taper half-width $\varepsilon$, which makes the reward continuous in the execution times. $72$ targets in $8$ clusters follow linear-Gaussian dynamics with continuous-time maneuvers, with process noise injected at round boundaries only. Each agent maintains its own belief over the target states, fused from raw neighbor measurements arriving under heterogeneous integer-round delays and an omniscient evaluator computes the reported metrics. Execution times are sampled one-sided as $\tau_i(t)=t\Delta+\mathcal{U}(0,\rho)$ with $\Delta$ the round duration, and we take $\tau_{\mathrm{ref}}(t)=t\Delta$.

\textbf{Reward.}
Let $\mathcal{M}$ index the targets, let $y_m\in\mathbb{R}^4$ be target $m$'s
position--velocity state, let $H=[\,I_2\;\;0\,]$ be the position selector, and let
$\Phi(\delta)$ be the intra-round state transition over an interval $\delta$, which is
deterministic and invertible because process noise is injected only at round boundaries.
An element $e=(a,\tau)$ with $a\in\calV_i$ denotes agent $i$ orienting its sector to heading
$\theta_a$ and measuring at instant $\tau$. Its measurement is of the target position at
$\tau$, so relative to the state at $\tau_{\mathrm{ref}}$ it acts through
\begin{equation}
\label{eq:obsmap}
G(\tau)\;\triangleq\;H\,\Phi(\tau-\tau_{\mathrm{ref}}),
\end{equation}
and contributes to target $m$ the information $\lambda_{e,m}\,G(\tau)^{\!\top}G(\tau)$, where
the scalar weight
\begin{equation}
\label{eq:lambda}
\lambda_{e,m}\;\triangleq\;
\frac{g\big(\beta_{i,m}(\tau)-\theta_a\big)}
{\sigma_0^{2}\big(1+\|p_i-q_m(\tau)\|^{\alpha}/d_0^{\alpha}\big)}
\end{equation}
combines the bearing gain $g(\cdot)\in[g_{\mathrm{out}},1]$, evaluated at the offset between
the target bearing $\beta_{i,m}(\tau)$ and the selected heading, with a range-dependent noise
variance of reference range $d_0$ and exponent $\alpha$. Writing $\Omega^-_m$ for the prior
information matrix of target $m$ at $\tau_{\mathrm{ref}}$ and
\begin{equation}
\label{eq:Lambda}
\Lambda_m(\mathcal{D})\;\triangleq\;
\sum_{e\in\mathcal{D}}\lambda_{e,m}\,G(\tau_e)^{\!\top}G(\tau_e),
\end{equation}
the round reward is the mutual information between the target states and the collected
measurements,
\begin{equation}
\label{eq:mireward}
F(\mathcal{D})\;\triangleq\;\frac{1}{2R_{\max}}\sum_{m\in\mathcal{M}}
\log\det\Big(I+(\Omega^-_m)^{-1}\Lambda_m(\mathcal{D})\Big),
\end{equation}
with the prior eigenvalues capped and $R_{\max}$ an analytic constant chosen so that
\Cref{asm:B} holds exactly; the measured clip rate is ${0}$. When every element shares a
common instant, $G=H$ and $\Lambda_m$ touches only the position block, so each summand
reduces to a $2\times2$ determinant. The normalization by $R_{\max}$ serves the analysis of
\Cref{sec:async}, which requires marginals in $[0,1]$. Since it is a fixed positive scaling
it affects neither the ordering of action sets nor the curvature, and we therefore report rewards below in their unnormalized units of nats.

\textbf{Relation to entropy-based objectives.}
The active information acquisition literature typically poses the sensing problem as
minimization of the entropy of the estimation task conditioned on the future
measurements~\cite{atanasov2015decentralized}. Within a round the two objectives coincide.
Since $I(y;z_{\mathcal{D}})=H(y)-H(y\mid z_{\mathcal{D}})$, and the prior entropy $H(y)$ is
fixed by the belief at $\tau_{\mathrm{ref}}$ and hence by the history rather than by the
round's actions, the two differ over a round by an additive constant. They therefore induce
the same ordering on deployment schedules and, more to the point, identical marginals,
\begin{equation}
\label{eq:mi-entropy}
F\big(e\mid\mathcal{D}\big)\;\propto\;
H\big(y\mid z_{\mathcal{D}}\big)-H\big(y\mid z_{\mathcal{D}\cup\{e\}}\big),
\end{equation}
so each agent's bandit reward is unaffected by the choice.

We nonetheless adopt the mutual information form, because the negative conditional entropy
is not normalized: $-H(y\mid z_{\varnothing})=-H(y)\neq0$ in general, which violates
\Cref{def:submodular}, places the per-agent rewards outside the interval $[0,1]$ required by
\Cref{prob:ABD}, and leaves the curvature of \Cref{def:curvature} undefined. Mutual
information satisfies $F(\varnothing)=0$ by construction and is the normalized
representative of the same objective.

\textbf{Baselines.}
Because \dog waits for a round's neighborhood information before acting on it, it defers
both its weight update \emph{and} the fusion of that round's measurements into its belief.
\alg uses partial information immediately in both, correcting as the remaining information
arrives. The two algorithms therefore differ by the consistent application of their
respective philosophies, rather than in the decision layer alone.

\textbf{The reward is an instance of \Cref{def:reward}.}
\Cref{def:reward} asks for a set function of the deployment schedule alone. $F$ takes
$\mathcal{D}$ and returns a number, with no separate argument saying when the reward is
evaluated. However, the reward
defined in~\eqref{eq:mireward} does not have that form. It asks how much the measurements reveal about the target state \emph{at $\tau_{\mathrm{ref}}$}, and a different
choice of $\tau_{\mathrm{ref}}$ asks about a different random variable. Taken at face value, the
formulation of \Cref{sec:async} would not apply to it. However, through the following lemma we show that the arguments, $\tau_{\mathrm{ref}}$ and $\mathcal{D}$, do not interact.
 
\begin{lemma}[Well-definedness]
\label{lem:welldef}
Let $\tau,\tau'$ lie in the same round and let $y(\tau')=\Phi(\tau'-\tau)\,y(\tau)$ with
$\Phi(\cdot)$ deterministic and invertible. Then
\begin{equation}
\label{eq:anchorinv}
I\big(y(\tau');z_{\mathcal{D}}\big)\;=\;I\big(y(\tau);z_{\mathcal{D}}\big)
\qquad\text{for every }\mathcal{D}\subseteq\mathcal{E}.
\end{equation}
\end{lemma}
 
\begin{proof}
Abbreviate $\Phi=\Phi(\tau'-\tau)$ and $z=z_{\mathcal{D}}$. Since $\Phi$ is deterministic
and invertible, the change of variables $y(\tau')=\Phi\,y(\tau)$ has constant Jacobian
determinant $|\det\Phi|$, so the differential entropies satisfy
\begin{equation}
h\big(y(\tau')\big)=h\big(y(\tau)\big)+\log|\det\Phi|,
\end{equation}
and the same shift applies conditionally on $z$, since $\Phi$ does not depend on $z$:
\begin{equation}
h\big(y(\tau')\mid z\big)=h\big(y(\tau)\mid z\big)+\log|\det\Phi|.
\end{equation}
Subtracting, the two $\log|\det\Phi|$ terms cancel and
\begin{equation}
I\big(y(\tau');z\big)=h\big(y(\tau')\big)-h\big(y(\tau')\mid z\big)
=I\big(y(\tau);z\big).
\end{equation}
\end{proof}
 
Process noise is injected only at round boundaries, so within a round the state propagates
deterministically through the invertible transition $\Phi$ and \Cref{lem:welldef} applies. Thus
the reference instant may be moved anywhere in the round without changing the value
of the reward as defined in \cref{eq:mireward}. The reward is therefore a function of $\mathcal{D}$ alone and is an
instance of \Cref{def:reward}, as the analysis of \Cref{sec:async} requires.
 
This has an important consequence. The
realized reward $F(\mathcal{D}_t)$ does not depend on the reference instant, but the reference
synchronous deployment $\bar{\mathcal{D}}_t=\{(a_{i,t},\tau_{\mathrm{ref}})\}_{i\in\mathcal{N}}$
is \emph{defined} by it. Choosing $\tau_{\mathrm{ref}}$ is choosing which synchronized
configuration to compare against. Magnitudes such as
$|F(\mathcal{D}_t)-F(\bar{\mathcal{D}}_t)|$, and every bound of \Cref{sec:async}, are
therefore insensitive to the choice, whereas the sign of that difference is not and is meaningful only relative to a named baseline.
 
\textbf{Submodularity.}
\Cref{asm:A} is likewise provable for this reward. Recall that
$F|_\tau$ is the restriction of $F$ to sets of action-instants that all carry the same timestamp $\tau$, so a slice is indexed by an ordinary set $A\subseteq\mathcal{V}$ of actions,
one per participating agent.
 
\begin{lemma}[Submodularity]
\label{lem:slice}
Fix an instant $\tau$, and let each action $a\in\mathcal{V}$, executed at that instant,
contribute to target $m$ the information $\lambda_{a,m}P_m$, where $\lambda_{a,m}\ge0$ is a
scalar depending on the action and the target, and $P_m\succeq0$ depends on neither and is
therefore shared by all actions. Then $F|_\tau$ is normalized, non-decreasing,
submodular, and $2$nd-order submodular on $2^{\mathcal{V}}$.
\end{lemma}
 
\begin{proof}
Fix a target $m$ and collect the contributions of a set $A\subseteq\mathcal{V}$. Because
every action contributes a multiple of the same matrix $P_m$, the total information is
\begin{equation}
\label{eq:scalarcollapse}
\sum_{a\in A}\lambda_{a,m}P_m=\lambda_{A,m}P_m,
\qquad
\lambda_{A,m}\triangleq\sum_{a\in A}\lambda_{a,m},
\end{equation}
so the entire dependence on $A$ is through the single non-negative number
$\lambda_{A,m}$, which is modular in $A$. Let $\mu_{m,1},\dots,\mu_{m,n}\ge0$ be the
eigenvalues of $(\Omega^-_m)^{-1}P_m$. Then
\begin{equation}
\label{eq:scalarform}
F|_\tau(A)=\sum_{m\in\mathcal{M}}g_m\big(\lambda_{A,m}\big),
\end{equation}
where
\begin{equation}
\label{eq:gdef}
g_m(x)\;\triangleq\;\frac{1}{2R_{\max}}\sum_{k}\log\big(1+x\,\mu_{m,k}\big),
\end{equation}
which reduces the claim to properties of the scalar functions $g_m$ on $x\ge0$:
\begin{equation}
\label{eq:gderivs}
g_m(0)=0,\quad g_m'>0,\quad g_m''<0,\quad g_m'''>0,
\end{equation}
each obtained by differentiating \cref{eq:gdef} termwise, the $k$th term
contributing $\mu_{m,k}/(1+x\mu_{m,k})$, $-\mu_{m,k}^2/(1+x\mu_{m,k})^2$ and
$2\mu_{m,k}^3/(1+x\mu_{m,k})^3$ respectively, up to the factor $1/(2R_{\max})$.
 
The first three properties in \cref{eq:gderivs} give the first three claims: $g_m(0)=0$
and $\lambda_{\varnothing,m}=0$ give normalization; $g_m'>0$ with $\lambda_{A,m}$
non-decreasing in $A$ gives monotonicity; and a concave non-decreasing function of a modular
set function is submodular.
 
For the fourth, fix an action $s$ and write $\lambda_s\triangleq\lambda_{s,m}$ and
\begin{equation}
h(x)\triangleq g_m(x+\lambda_s)-g_m(x),
\end{equation}
so that the marginal gain of $s$ against a set $C$ is $h(\lambda_{C,m})$. Writing
the second order submodular condition in terms of $h$, with $u=\lambda_{C,m}$,
$\alpha=\lambda_{A,m}$ and $\beta=\lambda_{B,m}$ for disjoint $A,B,C$, it becomes
\begin{equation}
h(u)-h(u+\alpha)\;\ge\;h(u+\beta)-h(u+\alpha+\beta),
\end{equation}
that is, the decrease of $h$ over an interval of length $\alpha$ is largest at the left, which
holds precisely when $h$ is convex. Finally
\begin{equation}
h''(x)=g_m''(x+\lambda_s)-g_m''(x)\;\ge\;0,
\end{equation}
since $g_m'''>0$ makes $g_m''$ non-decreasing and $\lambda_s\ge0$. Each summand
of \cref{eq:scalarform} therefore satisfies all four properties, and each is preserved
under summation over $m$.
\end{proof}

By \cref{eq:Lambda} the hypothesis of \Cref{lem:slice} holds with
$P_m=G(\tau)^{\!\top}G(\tau)$, which is common to all elements of a slice precisely because
they share the instant $\tau$; the bearing gain and range law enter only through the scalar
$\lambda_{e,m}$. This is what the sensing model is designed to guarantee, and it is the
reason accuracy is allowed to vary with bearing but not to become direction-dependent within
a single measurement. Across mixed timestamps the hypothesis fails, since
$G(\tau)\ne G(\tau')$, which is why \Cref{asm:A} is stated for perfectly synchronous execution times alone and why no result
of \Cref{sec:async} requires submodularity beyond them.

\textbf{Preliminary results.}
Empirical checks agree with the theory. Sampling from realized states, \Cref{asm:A} was violated in $0$ of ${39{,}000}$ trials and every marginal fell within $[0,1]$; the oscillation bound $\mathrm{osc}_1(\rho)$ was estimated directly at each $\rho$ and exhibits no jump at the origin, placing this reward in the first regime of \Cref{rem:dichotomy}; and the gaps of \Cref{lem:global,lem:marginal} were bounded by their predicted ceilings at every measured configuration, with roughly an order of magnitude of margin.

\Cref{fig:info_delay} reports mean reward against the maximum one-hop delay $\bar{d}$, over $n={20}$ seed-paired runs. The two algorithms agree at $\bar{d}=1$, after which \alg's reward remains essentially constant in $\bar{d}$ while \dog's declines monotonically. \alg improves on \dog by up to ${14.81}\%$ at $\bar{d}=30$. Acting on partial neighborhood information therefore absorbs almost the whole cost of delay on this task, and the advantage grows with delay as in \Cref{sec:sim-coverage}.

\Cref{fig:info_rho} isolates asynchrony by holding the delay fixed and sweeping the timing mismatch bound. Realized reward degrades smoothly and close to proportionally in $\rho$ relative to synchronous execution, over $n={10}$ seed-paired runs. The curves at $\bar{d}=0$ and $\bar{d}=10$ coincide to within fractions of a percent, which is what the analysis predicts. Communication delay enters the regret bound through $\bar{M}_T$ and clock mismatch through $\gamma$, additively and with no cross term. This is the graceful degradation asserted by \Cref{thm:async}, and it is the effect the coverage task of \Cref{sec:sim-coverage} cannot exhibit.

\section{Conclusion} \label{sec:con}


This paper introduces a distributed online optimization framework for submodular coordination under heterogeneous communication delays and asynchronous local clocks. The key capability provided by \alg is that agents can learn from partial neighborhood information as it arrives, instead of waiting for complete delayed feedback. This reduces the effective delay between acting and learning, enabling more timely coordination in dynamic environments while preserving provable network-level approximation guarantees.
The simulations validate our approach. \alg performs similarly to \dog under small delays, but increasingly outperforms \dog on larger delays by adapting earlier to changing conditions. Thus, the advantage of \alg is improved decision quality under delayed communication, rather than a fundamentally different asymptotic convergence rate.


 
Our future work will focus on leveraging adaptive bandit algorithms, 
such as Optimistic Hedge~\cite{rakhlin2013online}, 
to improve responsiveness to rapidly changing environments.



\bibliographystyle{IEEEtran}
\bibliography{references}

\appendices

\section{Proof of~\Cref{thm:per_agent}}\label{app:per_agnet_updated}


We prove the main result by establishing how far off $\hat{p}_t$ of \alg is from the reference distribution $p^{\text{Exp3}}_{t}$ corresponding to the standard \scenario{EXP3} bandit algorithm without delays (for the nonstochastic case)\cite{lattimore2020bandit}. To that end, we choose to work with losses instead of rewards and define the loss and the importance weighted loss estimate as:
\begin{align}
    l_{a,t} = 1-r_{a,t}, \quad 
    \tilde{l}_{a,t} = \frac{\mathbf{1}\{a=a_t\}}{\hat{p}_{a,t}}l_{a_t,t}.
\end{align}

We also define the cumulative loss up to round $t$ as
\begin{equation}
    \widetilde{L}_{a,t} = \sum_{s = 1}^t \tilde{l}_{a,t}.
\end{equation}

Now in our setting, the algorithm uses approximate per-round loss estimates $\hat{l}_{a,s}^{(t)}$ for each action $i$ and past round $s \leq t$ where

\begin{equation}
    \hat{l}^{(t)}_{a,t-\bar{d}} = \tilde{l}_{a,t-\bar{d}},
\end{equation}that is, our algorithm maintains accurate losses to rounds up to $t-\bar{d}$ where $\bar{d}$ is the maximum delay for agent $i$ receiving its neighbors' information. 

For action $a$, round $s$, and current round $t \geq s$, we define the per-round estimation error as

\begin{equation}
    e^{(t)}_{a,s} = \hat{l}^{(t)}_{a,s} - \tilde{l}^{(t)}_{a,s}
\end{equation}

where $e^{(t)}_{a,s} = 0$ for $s \leq t-\bar{d}$. We also define the loss formulation equivalent of the cumulative error (\cref{eq:cumulative_err_def}) as

\begin{equation}
    \varepsilon_{a,t} = \hat{L}_{a,t} - \widetilde{L}_{a,t} = \sum_{s=t-\bar{d}+1}^t e^{(t)}_{a,s}. \label{eq:loss_difference}
\end{equation}

Now, we recall the definition of probability distributions for both \alg and the reference as
\begin{align}
\hat{p}_{a,t}
&= \frac{\exp(\theta_a)}{\sum_k \exp(\theta_k)},
&
p_{a,t}^{\text{Exp3}}
&= \frac{\exp(\theta_a')}{\sum_k \exp(\theta_k')},
\label{eq:prob_softmax_def}
\end{align}where $\theta_a = -\eta \hat{L}_{a,t-1}$ and $\theta_a' = -\eta \widetilde{L}_{a,t-1}$.

We also know that the softmax function has a $1/2$-lipschitz bound irrespective of the $l_p$ norm \cite{nair2026softmax}, then we have
\begin{equation}
    || \sigma(x) - \sigma(y) ||_1 \leq \frac{1}{2}||x-y||_1.
\end{equation}
Combining this inequality with \cref{eq:loss_difference,eq:prob_softmax_def} results in
\begin{align}
    ||\hat{p}_t(\theta) - p_{t}^{\text{Exp3}}(\theta')||_1 &\leq \frac{1}{2} ||\eta\hat{L}_{t-1} - \eta\widetilde{L}_{t-1}||_1
    \\
    &\leq \frac{\eta}{2} \sum_a |\varepsilon_{a,t-1}|. \label{eq:pt_diff_bound_new}
\end{align}To bound the term above we define
\begin{equation}
    M_t \triangleq \max_{a \in \{1,\ldots,|\mathcal{V}|\}} \left|\varepsilon_{a,t}\right| = \max_a \left|\hat{L}_{a,t} - \widetilde{L}_{a,t}\right|, \label{eq:Mt_def}
\end{equation}
which is equivalent to the reward based definition of the maximum cumulative loss in \cref{eq:Mt_definition}. Now we can further bound the expectation of $|\varepsilon_{a,t}|$ in the worst case, by assuming all estimates of losses are as far from the truth as possible, 
\begin{align}
    \mathbb{E}[|\varepsilon_{a,t}|] &= \mathbb{E}\left[\left|\sum_{s=t-\bar{d}+1}^t e^{(t)}_{a,s}\right|\right]
    \\
    &\leq \sum_{s=t-\bar{d}+1}^t \mathbb{E}[|e^{(t)}_{a,s}|]
    \\
    &\leq \sum_{s=t-\bar{d}+1}^t \mathbb{E}\left[\frac{|\hat{l}^{\text{raw},(t)}_{a,s} - l_{a,s}|}{\hat{p}_{a,s}} \mathbf{1}\{a = a_s\}\right]
    \\
    &\leq \sum_{s=t-\bar{d}+1}^t \mathbb{E}\left[\frac{\mathbf{1}\{a = a_s\}}{\hat{p}_{a,s}}\right] \label{eq:worst_bound_loss_diff}
    \\
    &= \sum_{s=t-\bar{d}+1}^t \mathbb{E}\left[\mathbb{E}\left[\frac{\mathbf{1}\{a = a_s\}}{\hat{p}_{a,s}} \middle|\, \mathcal F_{s-1} \right]\right] = \bar{d}, \label{eq:expectation_of_identity}
\end{align}where \cref{eq:worst_bound_loss_diff} comes from the worst case bound due to $l_{a,s},\hat{l}^{\text{raw}}_{a,s} \in [0,1]$ and \cref{eq:expectation_of_identity} results from applying the law of total expectation and $\mathbb{E}[\mathbf{1}\{a = a_s\}] = \hat{p}_{a,s}$ with $\mathcal F_{s-1}$ being the $\sigma$-algebra of all information available to the agent up to round $s-1$. We can now bound $\mathbb{E}[M_t]$ in the worst case as
\begin{equation}
\label{eq:Mt_bound_worst_case}
    \mathbb{E}[M_t] = \mathbb{E}\left[\max_{a \in \mathcal{V}} \,|\varepsilon_{a,t}|\right] \leq \sum_{a = 1}^{|\mathcal{V}|} \mathbb{E}[|\varepsilon_{a,t}|] \leq |\mathcal{V}|\bar{d}.
\end{equation}Now we express the per-agent regret as
\begin{equation}
\mathrm{Reg}_T
=
\sum_{t=1}^T l_{a_{t},t} - \min_{a\in\mathcal V}\sum_{t=1}^T l_{a,t}.
\end{equation}
Taking expectation and using
$
l_{a_{t},t}
=
\sum_{a\in\mathcal V}\mathbf 1\{a_{t}=a\}\,l_{a,t},
$ 
together with the law of total expectation yields
\begin{equation}
\begin{aligned}
\mathbb E[\mathrm{Reg}_T]
&=
\sum_{t=1}^T \mathbb E[l_{a_{t},t}] 
- \min_{a\in\mathcal V}\sum_{t=1}^T l_{a,t} \\
&\hspace{-3em}=
\sum_{t=1}^T
\mathbb E\!\left[
\mathbb E\!\left[
\sum_{a\in\mathcal V}\mathbf 1\{a_{t}=a\}l_{a,t}
\,\middle|\, \mathcal F_{t-1}
\right]\right] - \min_{a\in\mathcal V}\sum_{t=1}^T l_{a,t}\\
&\hspace{-3em}=
\sum_{t=1}^T
\mathbb E\!\left[\sum_{a\in\mathcal V} p_{a,t} l_{a,t}\right] -\min_{a\in\mathcal V}\sum_{t=1}^T l_{a,t}, \label{eq:expectation_regret_new}
\end{aligned}
\end{equation}
where $p_{a,t}=\mathbb P(a_{t}=a\mid \mathcal F_{t-1})$ and $\mathcal F_{t-1}$ is the $\sigma$-algebra of all information available to the agent up to round $t-1$. 

Taking the difference between the expected regret of \alg's and \scenario{Exp3} and applying \cref{eq:expectation_regret_new} results in
\begin{align}
&\mathbb{E}[\text{Reg}_T - \text{Reg}_T^{\text{Exp3}}] = \mathbb{E}[\text{Reg}_T]
   - \mathbb{E}[\text{Reg}_T^{\text{Exp3}}]
   \nonumber \\
& =
\begin{aligned}
&\sum_{t=1}^{T} \mathbb{E}\Big[
    \sum_{a \in \mathcal{V}} \hat{p}_{a,t} \, l_{a,t}
  \Big]
- \sum_{t=1}^{T} \mathbb{E}\Big[
    \sum_{a \in \mathcal{V}} p_{a,t}^{\text{Exp3}} \, l_{a,t}
  \Big]
\end{aligned}
\label{eq:reg_diff_2_n} \\
& = \sum_{t=1}^{T} \mathbb{E}\Big[
      \big\langle \hat{p}_t - p^{\text{Exp3}}_t,\, l_t \big\rangle
   \Big]
   \label{eq:reg_diff_3_n} \\
& \le \sum_{t=1}^{T} \mathbb{E}\Big[
      \|\hat{p}_t - p^{\text{Exp3}}_t\|_1
   \Big]
   \label{eq:reg_diff_4_n} \\
& \le \frac{1}{2} \sum_{t=1}^{T} \frac{\eta |\mathcal{V}|}{2}  \mathbb{E}\left[
      M_{t-1}
   \right]
   \label{eq:reg_diff_5_n}
\end{align}where we used $r_{a,t} \in [0,1]$ for all $a,t$ to obtain equation \cref{eq:reg_diff_4_n} and used equation \cref{eq:pt_diff_bound_new} and \cref{eq:Mt_def} to obtain equation \cref{eq:reg_diff_5_n}. Substituting the regret bound of the standard Exp3 without delays from \cite{lattimore2020bandit}, we the following regret bound for \alg
\begin{equation}
    \mathbb{E}[\text{Reg}_T] \leq \frac{\ln(|\mathcal{V}|)}{\eta} + \eta |\mathcal{V}| T + \frac{\eta |\mathcal{V}| T}{4} \bar{M}_T,
\end{equation}where $\bar{M}_T$ is defined in \cref{eq:Mt_bar_definition}. With a learning rate of $\eta = \sqrt{\frac{\ln |\mathcal{V}|}{|\mathcal{V}| T (1+\bar{M}_T/4)}}$, we have an average expected regret of
\begin{equation}
    \frac{\mathbb{E}[\mathrm{Reg}_T]}{T} \leq O\!\left(\left(\sqrt{|\mathcal V|(1+\bar{M}_T/4)} \right) \sqrt{\frac{\ln|\mathcal V|}{T}} \right).
\end{equation}Substituting the worst case bound of $\mathbb{E}[M_t]$ from \cref{eq:Mt_bound_worst_case} and using a learning rate of $\eta = \sqrt{\frac{\ln |\mathcal{V}|}{|\mathcal{V}| T (1+\bar{d}/4)}}$ results in an average expected regret of 
\begin{equation}
    \frac{\mathbb{E}[\mathrm{Reg}_T]}{T} \leq O\!\left(\left(\sqrt{|\mathcal V|(1+|\mathcal V|\bar{d}/4)} \right) \sqrt{\frac{\ln|\mathcal V|}{T}} \right),
\end{equation}
completing the proof the theorem.

\section{Proof of~\Cref{thm:approx_performance}}\label{app:main}

We prove the result in \Cref{thm:approx_performance} as follows: %
{\begin{align}
    &\sum_{t=1}^{T} f_t(\calA^\opt)\nonumber\\
    &=\sum_{t=1}^{T} f_t(\calA^\opt\cup\calA_t) - \sum_{t=1}^{T} \sum_{i\in\calN} f_t(a_{i,t}\,|\,\calA^\opt\cup\{a_{j,t}\}_{j\in[i-1]}) \label{aux22:1}\\
    &\leq\sum_{t=1}^{T} f_t(\calA_t) + \sum_{t=1}^{T} \sum_{i\in\calN} f_t(a_{i}^\opt\,|\,\calA_t) \nonumber\\
    &\quad - (1-\kappa_{f}) \sum_{t=1}^{T} \sum_{i\in\calN} f_t(a_{i,t}\,|\,\{a_{j,t}\}_{j\in\calN_{i}}) \label{aux22:2}\\
    &\leq\sum_{t=1}^{T} f_t(\calA_t) + \kappa_{f} \sum_{t=1}^{T} \sum_{i\in\calN} f_t(a_{i,t}\,|\,\{a_{j,t}\}_{j\in\calN_{i}}) \nonumber\\
    &\quad + \sum_{i\in\calN} \sum_{t=1}^{T} \mleft[f_t(a_{i}^\opt\,|\,\{a_{j,t}\}_{j\in\calN_{i}}) - f_t(a_{i,t}\,|\,\{a_{j,t}\}_{j\in\calN_{i}})\mright] \label{aux22:3}\\  
    &\leq\sum_{t=1}^{T} f_t(\calA_t) + \sum_{i\in\calN} \Reg(\{a_{i,t}\}_{t\in [T]}) \nonumber\\
    &\quad + \kappa_{f} \sum_{t=1}^{T} \sum_{i\in\calN} f_t(a_{i,t}\,|\,\{a_{j,t}\}_{j\in\calN_{i}}) \label{aux22:4}\\
    &=(1+\kappa_f) \sum_{t=1}^{T} f_t(\calA_t) + \sum_{i\in\calN} \Reg(\{a_{i,t}\}_{t\in [T]}) \nonumber\\
    &\quad+ \kappa_{f} \sum_{t=1}^{T} \sum_{i\in\calN} \mleft[f_t(a_{i,t}\,|\,\{a_{j,t}\}_{j\in\calN_{i}}) - f_t(a_{i,t}\,|\,\{a_{j,t}\}_{j\in [i-1]})\mright] \label{aux22:5}\\
    &\leq(1+\kappa_f) \sum_{t=1}^{T} f_t(\calA_t) + \sum_{i\in\calN} \Reg(\{a_{i,t}\}_{t\in [T]}) \nonumber\\
    &\quad+ \kappa_{f} \sum_{t=1}^{T} \sum_{i\in\calN} \mleft[f_t(a_{i,t}) - f_t(a_{i,t}\,|\,\{a_{j,t}\}_{j\in [i-1]\setminus\calN_{i}})\mright] \label{aux22:6}\\
    &\leq (1+\kappa_f) \sum_{t=1}^{T} f_t(\calA_t) + \sum_{i\in\calN} \Reg(\{a_{i,t}\}_{t\in [T]}) \nonumber\\
    &\quad+ \kappa_{f} \sum_{t=1}^{T} \sum_{i\in\calN} \underbrace{\mleft[f_t(a_{i,t}) - f_t(a_{i,t}\,|\,\{a_{j,t}\}_{j\in \calN_{i}^c})\mright]}_{\ourcurv_{f_t,i}(\calN_{i})}, \label{aux22:7}
\end{align}}where \cref{aux22:1} holds by telescoping the sum, \cref{aux22:2} holds since $f$ is submodular and since $1-\kappa_{f} \leq \frac{f_t(a_{i,t}\,|\,\{a_{j,t}\}_{j\in\calN\setminus\{i\}})}{f_t(a_{i,t})} \leq \frac{f_t(a_{i,t}\,|\,\calA^\opt\cup\{a_{j,t}\}_{j\in[i-1]})}{f_t(a_{i,t}\,|\,\{a_{j,t}\}_{j\in\calN_{i}})}$ per \Cref{def:curvature}, \cref{aux22:3} holds from submodularity, \cref{aux22:4} holds from \Cref{eq:static_regret_agent_i}, \cref{aux22:6} holds since $f_t$ is 2nd-order submodular, and \cref{aux22:7} holds from \Cref{def:coin}. %

Taking the expectation of both sides of \cref{aux22:7} and leveraging \Cref{thm:per_agent}, we prove \cref{eq:thm1_intermediate} by the following,
\begin{align}
    &\mathbb{E}\mleft[\sum_{t=1}^T f_t(\solopt)\mright] \nonumber \\
    &\leq (1+\kappa_f)\, \mathbb{E}\mleft[\sum_{t=1}^T f_t(\calA_t)\mright] + \kappa_f \sum_{i\in\calN} \mathbb{E}\mleft[\sum_{t=1}^T \ourcurv_{f_t,i}(\calN_{i})\mright] \nonumber \\\label{aux22:10}
    &\quad + \tilde{O}\mleft(|\calN|\sqrt{|\bar{\mathcal V}|(1+\mathbb{E}[M_t]/4)T} \mright).
\end{align}

In the fully centralized scenario, we have $\calN_{i}=\calN\setminus\{i\}$. Thus, $\ourcurv_{f_t,i}(\calN_{i})= 0$, and thus \cref{eq:thm1_centralized} is proved.

Finally, in the fully decentralized case where $\calN_{i}=\emptyset$, taking expectations on both sides of \cref{aux22:4}, we have

\begin{align}
    &\mathbb{E}\mleft[\sum_{t=1}^T f_t(\calA_t)\mright] \nonumber\\
    &\geq \mathbb{E}\mleft[\sum_{t=1}^T f_t(\solopt)\mright] - \kappa_f \sum_{i\in\calN} \mathbb{E}\mleft[\sum_{t=1}^T f_t(a_{i,t})\mright] \nonumber\\
    &\quad - \tilde{O}\mleft(|\calN|\sqrt{|\bar{\mathcal V}|(1+\mathbb{E}[M_t]/4)T} \mright)
    \nonumber\\
    &\geq \mathbb{E}\mleft[\sum_{t=1}^T f_t(\solopt)\mright] - \frac{\kappa_f}{1-\kappa_f} \sum_{i\in\calN} \mathbb{E}\mleft[\sum_{t=1}^T f_t(\calA_t)\mright] \nonumber\\
    &\quad - \tilde{O}\mleft(|\calN|\sqrt{|\bar{\mathcal V}|(1+\mathbb{E}[M_t]/4)T} \mright).
\end{align}and thus \cref{eq:thm1_decentralized} is proved. \qed

\section*{Appendix III: Proofs of Section~\ref{sec:async}}

Throughout, fix a round $t$, write $M\triangleq|\mathcal{N}|$, and abbreviate
$\tau_{\mathrm{ref}}\triangleq\tau_{\mathrm{ref}}(t)$. For each agent $j$ let
\begin{equation}
\label{eq:appshort}
e_j\triangleq\big(a_{j,t},\tau_j(t)\big),
\qquad
\bar e_j\triangleq\big(a_{j,t},\tau_{\mathrm{ref}}\big)
\end{equation}
denote its realized action-instant and the synchronous counterpart of that instant, so that
$\mathcal{D}_t=\{e_j\}_{j=1}^{M}$ and $\bar{\mathcal{D}}_t=\{\bar e_j\}_{j=1}^{M}$. By
\cref{eq:bounded_skew}, $e_j$ and $\bar e_j$ carry the same action and timestamps differing by at
most $\rho$, for every $j$.

\subsection*{Proof of Lemma~\ref{lem:global}}

Index the agents $1,\dots,M$ in any fixed order and define the hybrid schedules
\begin{equation}
\mathcal{H}_k\;\triangleq\;\{\bar e_j\}_{j\le k}\cup\{e_j\}_{j>k},
\qquad k=0,\dots,M,
\end{equation}
so that $\mathcal{H}_0=\mathcal{D}_t$ and $\mathcal{H}_M=\bar{\mathcal{D}}_t$. Consecutive
hybrids differ in exactly one element, that of agent $k$. Writing
\begin{equation}
S_k\;\triangleq\;\{\bar e_j\}_{j<k}\cup\{e_j\}_{j>k}
\end{equation}
for their common part, we have $\mathcal{H}_{k-1}=S_k\cup\{e_k\}$ and
$\mathcal{H}_k=S_k\cup\{\bar e_k\}$, so the term $F(S_k)$ cancels and
\begin{equation}
\label{eq:hybridstep}
F(\mathcal{H}_{k-1})-F(\mathcal{H}_k)
=F\big(e_k\mid S_k\big)-F\big(\bar e_k\mid S_k\big).
\end{equation}
The right-hand side compares two marginals under a common context whose added elements
differ only in their timestamp, by at most $\rho$; it is therefore bounded in absolute value
by $\mathrm{osc}_1(\rho)$ by~\eqref{eq:osc1}. Telescoping over $k$ and applying the triangle
inequality,
\begin{align}
\big|F(\mathcal{D}_t)-F(\bar{\mathcal{D}}_t)\big|
&=\Big|\sum_{k=1}^{M}\big[F(\mathcal{H}_{k-1})-F(\mathcal{H}_k)\big]\Big|\notag\\
&\le M\,\mathrm{osc}_1(\rho).
\end{align}

\subsection*{Proof of Lemma~\ref{lem:marginal}}

Fix $a\in\mathcal{V}_i$ and write
\begin{equation}
e\triangleq(a,\tau_i(t)),\quad
\bar e\triangleq(a,\tau_{\mathrm{ref}}),
\end{equation}
\begin{equation}
S\triangleq\{e_j\}_{j\in\mathcal{N}_i},\quad
\bar S\triangleq\{\bar e_j\}_{j\in\mathcal{N}_i},
\end{equation}
so that $r^{\mathrm{async}}_{i,t}(a)=F(e\mid S)$ and
$r^{\mathrm{sync}}_{i,t}(a)=F(\bar e\mid\bar S)$. Inserting the intermediate term
$F(\bar e\mid S)$,
\begin{align}
\label{eq:margsplit}
r^{\mathrm{async}}_{i,t}(a)-r^{\mathrm{sync}}_{i,t}(a)
&=\big[F(e\mid S)-F(\bar e\mid S)\big]\notag\\
&\quad+\big[F(\bar e\mid S)-F(\bar e\mid\bar S)\big].
\end{align}
The first bracket compares two marginals of the same context $S$ whose added elements differ
only in timestamp, with $|\tau_i(t)-\tau_{\mathrm{ref}}|\le\rho$. It is
bounded by $\mathrm{osc}_1(\rho)$ by~\eqref{eq:osc1}. The second compares the marginal of the
fixed element $\bar e$ against two contexts satisfying $S\sim_\rho\bar S$, both of size
exactly $|\mathcal{N}_i|$. It is bounded by $\mathrm{osc}_2(\rho;|\mathcal{N}_i|)$
by~\eqref{eq:osc2}. The triangle inequality applied to~\eqref{eq:margsplit} gives
\begin{equation}
\big|r^{\mathrm{async}}_{i,t}(a)-r^{\mathrm{sync}}_{i,t}(a)\big|
\le\mathrm{osc}_1(\rho)+\mathrm{osc}_2\big(\rho;|\mathcal{N}_i|\big).
\end{equation}



\subsection*{Proof of Lemma~\ref{lem:conversion}}

Write $G\triangleq\sum_{t=1}^{T}\gamma_{i,t}$ and let
\begin{equation}
a^\star\in\arg\max_{a\in\mathcal{V}_i}\ \sum_{t=1}^{T}r^{\mathrm{sync}}_{i,t}(a).
\end{equation}
Since $\gamma_{i,t}$ bounds the discrepancy uniformly over $\mathcal{V}_i$, it applies at
$a^\star$, giving
\begin{align}
\sum_{t}r^{\mathrm{sync}}_{i,t}(a^\star)
&\le\sum_{t}r^{\mathrm{async}}_{i,t}(a^\star)+G\notag\\
&\le\max_{a\in\mathcal{V}_i}\sum_{t}r^{\mathrm{async}}_{i,t}(a)+G .
\end{align}
It applies equally at the played action $a_{i,t}\in\mathcal{V}_i$, giving
\begin{equation}
-\sum_{t}r^{\mathrm{sync}}_{i,t}(a_{i,t})
\le-\sum_{t}r^{\mathrm{async}}_{i,t}(a_{i,t})+G .
\end{equation}
Adding the two displays yields
$\mathrm{Reg}^{\mathrm{sync}}_T\le\mathrm{Reg}^{\mathrm{async}}_T+2G$, and $G\le T\gamma$.
The two uses are in opposite directions, which is why $\gamma_{i,t}$ is defined through an
absolute value and why the factor $2$ appears.

\subsection*{Proof of \Cref{thm:async}}

The chain establishing \Cref{thm:approx_performance} reaches \cref{aux22:7} and
rearranges to a bound on $\mathbb{E}[f_t(\mathcal{A}_t)]$ in which the per-agent regret
enters as
\begin{equation}
\frac{1}{T(1+\kappa_f)}\sum_{i\in\mathcal{N}}\mathrm{Reg}_T ,
\end{equation}
with $\mathrm{Reg}_T$ that of \cref{eq:static_regret_agent_i}, i.e.\
$\mathrm{Reg}^{\mathrm{sync}}_T$, whereas Theorem~\ref{thm:per_agent} bounds
$\mathrm{Reg}^{\mathrm{async}}_T$. Applying Lemma~\ref{lem:conversion} for each agent and
then Lemma~\ref{lem:marginal},
\begin{align}
\frac{1}{T(1+\kappa_f)}\sum_{i\in\mathcal{N}}2T\gamma_i
&\le\frac{2}{1+\kappa_f}\sum_{i\in\mathcal{N}}
\Big[\mathrm{osc}_1(\rho)\notag\\
&\qquad\qquad+\mathrm{osc}_2\big(\rho;|\mathcal{N}_i|\big)\Big]\notag\\
&\le\frac{2|\mathcal{N}|}{1+\kappa_f}
\big(\mathrm{osc}_1+\mathrm{osc}_2\big)(\rho).
\end{align}
This yields a guarantee on $\mathbb{E}[f_t(\mathcal{A}_t)]$, the synchronous value of the
played actions. The quantity earned by the system is $F(\mathcal{D}_t)$, so applying
Lemma~\ref{lem:global} once more,
\begin{equation}
\mathbb{E}\big[F(\mathcal{D}_t)\big]
\ge\mathbb{E}\big[f_t(\mathcal{A}_t)\big]-|\mathcal{N}|\,\mathrm{osc}_1(\rho),
\end{equation}
and~\eqref{eq:Gamma} follows. The final inequality of~\eqref{eq:Gamma} uses
$\kappa_f\in[0,1]$. In the fully decentralized case $\mathcal{N}_i=\varnothing$, so
$\mathrm{osc}_2(\rho;0)=0$ and the constant improves. \qed

\end{document}